\documentclass[smallcondensed]{svjour3}

\usepackage[utf8]{inputenc} 

\usepackage{geometry} 
\usepackage{graphicx} 
\usepackage{subfig} 
\usepackage{placeins} 
\usepackage{booktabs} 
\usepackage{multirow} 
\usepackage{xcolor}

\usepackage{paralist} 
\usepackage{verbatim} 

\usepackage{array} 
\usepackage{mathtools} 
\usepackage{amssymb} 

\usepackage{algorithmic} 
\usepackage[linesnumbered,noend]{algorithm2e}

\usepackage[hidelinks]{hyperref} 

\graphicspath{{./}}

\title{Using Subobservers to Synthesize Opacity-Enforcing Supervisors}
\author{Richard~Hugh~Moulton\thanks{RHM is the corresponding author and can be reached at richard.moulton@queensu.ca} \and Behnam~Behinaein~Hamgini \and Zahra~Abedi~Khouzani \and R{\^{o}}mulo~Meira-G{\'{o}}es \and Fei~Wang \and Karen~Rudie}
\institute{Richard Hugh Moulton \and Behnam Behinaein Hamgini \and Zahra Abedi Khouzani \at Department of Electrical and Computer Engineering, Queen's University, Kingston, Canada
\and
R{\^{o}}mulo Meira-G{\'{o}}es \at Department of Electrical Engineering and Computer Science, University of Michigan, Ann Arbor, USA
\and
Fei Wang \at College of Information Science and Engineering, Huaqiao University, Xiamen, China
\and
Karen Rudie \at Department of Electrical and Computer Engineering, Queen's University, Kingston, Canada\\Ingenuity Labs Research Institute, Queen's University, Kingston, Canada
}

\begin{document}
\maketitle

\begin{abstract}
In discrete-event system control, the worst-case time complexity for computing a system's observer is exponential in the number of that system's states. This results in practical difficulties since some problems require calculating multiple observers for a changing system, e.g.,\ synthesizing an opacity-enforcing supervisor. Although calculating these observers in an iterative manner allows us to synthesize an opacity-enforcing supervisor and although methods have been proposed to reduce the computational demands, room exists for a practical and intuitive solution. Here we extend the subautomaton relationship to the notion of a subobserver and demonstrate its use in reducing the computations required for iterated observer calculations. We then demonstrate the subobserver relationship's power by simplifying state-of-the-art synthesis approaches for opacity-enforcing supervisors under realistic assumptions.
\end{abstract}


\section{Introduction}
Discrete-event systems (DES) research has recently focused on opacity, a system property that ensures that secret information cannot be distinguished from non-secret information. The literature includes methods for verifying if a system is opaque, namely, whether a hostile agent could—based on its observations—determine whether or not the system is in a secret state, or whether or not the system has generated a secret event sequence. A natural question to ask if a system is not opaque is “how can we make this system opaque?” One approach is to alter which events are visible to an adversary, however this is rarely in the power of the system designer. An alternative approach is to use supervisory control to disable event occurrences so that the controlled system is opaque~\cite{Jacob2016}.

Researchers have examined opacity-enforcing supervisors under various assumptions~\cite{Badouel2007,Ben-Kalefa2011,Dubreil2010,Takai2008,Tong2018a,Yin2015,Yin2016a}. The na{\"{i}}ve method of producing such a supervisor starts by constructing a plant automaton representing the system and an observer automaton capturing the adversary's view of the system based on which events the adversary can observe. The standard method for constructing an observer automaton relies on converting a nondeterministic finite automaton to a deterministic finite automaton (NFA-to-DFA conversion)~\cite{Hopcroft1979}, which is known to require asymptotic computing time that is exponential in the number of states of the plant. Although supervisory control can then restrict plant behaviour to enforce opacity, a control-aware adversary will then be able to revise its plant estimates. As a result, repeated plant evolution necessitates repeated NFA-to-DFA operations. In this paper we present a method for synthesizing an opacity-enforcing supervisor that relies on only one NFA-to-DFA operation on the states of the uncontrolled plant.

Although we motivate our approach with the problem of plant evolution due to supervisory control, it is more generally a method for tracking the joint behaviour of a plant and an observer through multiple steps of plant evolution, without computing an observer at each intermediary stage. In addition supervisory control, a plant may evolve over time for other reasons as well. For example, a discrete-event process may be a naturally time-varying dynamic discrete-event system~\cite{Grigorov2006}, may require online control~\cite{Chung1992}, or may be best modelled by a time-varying automaton because more becomes known about the process over time. Alternatively, a discrete-event process may be controlled by a number of decentralized agents with different views of the plant,~\cite{Lin1988a}, where controls enacted by one agent may allow another agent to update its own estimates of the plant.

The power of our approach is that it leads to efficient and intuitive computations. In the remainder of this paper we review previous approaches to the opacity control problem, we introduce the subobserver property and show that it can be used to incrementally refine the joint behaviour of a plant and an observer, and, finally, we demonstrate how this approach can be used to efficiently synthesize an opacity-enforcing supervisor.

\section{Problem Definition and Related Works}
To lay the theoretical foundation for our contribution, we introduce the formalisms of DES, define the opacity control problem, and review the literature for related concepts and methods.

\subsection{Discrete-event systems}
DES are used to model processes that are discrete, asynchronous and potentially nondeterministic~\cite{Ramadge1987}. The discrete nature of these systems is captured with an automaton: at any given time the system is in a specific state and the system moves between states via event-based transitions.

Mathematically, a DFA $G$ is a 5-tuple $(Q, \Sigma, \delta_G, q_0, Q_m)$ where $Q$ is the finite set of states in the system, $\Sigma$ is the finite alphabet of events that can occur, $\delta_G: Q \times \Sigma \to Q$ is a partial function, $q_0 \in Q$ is the initial state, and $Q_m \subseteq Q$ is the set of marked states. From a control theory perspective the automaton $G$ represents the plant, or underlying process, to be controlled. The alphabet $\Sigma$ represents the events that can occur within the plant and can be divided into disjoint sets by considering whether individual events are controllable or uncontrollable ($\Sigma_c$ and $\Sigma_{uc}$ respectively) and whether they are observable ($\Sigma_o$ and $\Sigma_{uo}$ respectively). The transition function $\delta_G$ leads inductively from single events to the language $L(G)$, which is the plant's full range of behaviour. Because $\delta_G$ is usually a partial function, we use the notation $\delta_G(q,\sigma)!$ to denote that a specific transition $\delta_G(q,\sigma)$ is defined. We will also use the notation $\langle q', \sigma, q \rangle_G$ to denote that $\delta_G(q',\sigma) = q$. The set of marked states may represent the completion of a process or, in our case, a set of distinctive states that we wish to keep indistinguishable to an adversary. In this paper we call these ``secret states.''

A particular DFA representation of a language induces a binary relation called the equiresponse relation (Definition~\ref{def:equiresponseRelation}, from Lin and Wonham~\cite{Lin1988}) on $\Sigma^*$, the Kleene closure of $\Sigma$. The equiresponse relation for a DFA associates strings with the states they lead to, so that two strings are related if they lead to the same state. Because the equiresponse relation is a right congruence on $\Sigma^*$, it is also an equivalence relation and, like all equivalence relations, defines a partition. We note that different DFA representations of a language will induce different equiresponse relations and there is no canonical equiresponse relation.
\begin{definition}[Equiresponse relation] \label{def:equiresponseRelation}
Given an automaton $G = (Q, \Sigma, \delta_G, q_0, Q_m)$, the \textbf{equiresponse relation} of $G$ is a right congruence $eq(G)$ on $\Sigma^*$ and defined by $$s \equiv s' \mod eq(G) \iff \delta_G(q_0,s) = \delta_G(q_0, s')$$ and if $\delta_G(q_0,s)$ is undefined then $\delta_G(q_0,s) = \delta_G(q_0,s')$ if and only if $\delta_G(q_0,s')$ is undefined as well.
\end{definition}

\subsubsection{Observing discrete-event systems} \label{sec:observingDES}
When an automaton contains unobservable events, an observer automaton can be used to capture an agent's beliefs about the system. This observer automaton can be computed by transforming the original automaton, first by substituting the empty string $\epsilon$ for all transitions whose events are not in $\Sigma_o$~\cite[p.~76]{Cassandras2008}, and then performing NFA-to-DFA conversion~\cite[pp.~87-90]{Cassandras2008}. Note that we do not use the standard marking associated with NFA-to-DFA conversion (whereby a subset labelling a state in the DFA is marked if any element in the subset is marked in the original NFA); rather, we mark a state in the DFA only if {\em all\/} elements of the subset label are marked states in the original NFA. As will be discussed in Section~\ref{sec:proposedAlgorithm}, this way of marking states in the observer automaton will distinguish states in which an adversary can be sure that the original system is in a secret state.

Throughout this paper we denote the transformation of an automaton to an observer automaton with respect to the alphabet $\Sigma_o$ by $T_{\Sigma_o}(\cdot)$; we denote the associated change in languages generated by these automata as the projection $P_{\Sigma_o}(\cdot)$ (Definition~\ref{def:projectionString}, from Cassandras and Lafortune~\cite[p.~57]{Cassandras2008}). 
\begin{definition}[Projection of strings] \label{def:projectionString}
For an alphabet $\Sigma$ and another alphabet $\Sigma_o \subseteq \Sigma$, we define the \textbf{projection} $P_{\Sigma_o}: \Sigma^* \rightarrow \Sigma_o^*$ recursively as follows:
\begin{align*}
P_{\Sigma_o}(\epsilon) &= \epsilon\\
P_{\Sigma_o}(\sigma) &= \begin{cases}\sigma\quad&\text{if } \sigma \in \Sigma_o\\ \epsilon\quad&\text{if } \sigma \in \Sigma \setminus \Sigma_o\end{cases}\\
P_{\Sigma_o}(s\sigma) &= P_{\Sigma_o}(s)P_{\Sigma_o}(\sigma)\quad \text{for }s \in \Sigma^*, \sigma \in \Sigma
\end{align*}
\end{definition}

The projection operation implicitly defines an equivalence relation over $\Sigma^*$, where two strings $s$ and $s'$ are related if they share a projection, $P_{\Sigma_o}(s) = P_{\Sigma_o}(s')$. Each cell in the associated partition contains those strings whose projections are equal. Different projections will define different partitions on $\Sigma^*$ and we say that one partition refines another partition if it continues to distinguish between all elements that the latter partition distinguishes between (Definition~\ref{def:refinementOfPartition})~\cite[Ch.~9]{Rosen2019}.
\begin{definition}[Refinement of a partition] \label{def:refinementOfPartition}
Let $\rho_1$ and $\rho_2$ be binary relations on a set $X$. We say that $\rho_1$ \textbf{refines} $\rho_2$, or that $\rho_2$ is coarser than $\rho_1$, denoted $\rho_1 \leq \rho_2$ if $$(\forall x,y \in X)\ (x,y) \in \rho_1 \implies (x,y) \in \rho_2.$$
\end{definition}

\subsubsection{Controlling discrete-event systems}
In addition to observing DES processes, we can also control them. In DES the controller that modifies which events can occur in a plant $G$ is called the \emph{supervisor} and is formally a pair $\mathcal{S} = (S,\phi)$. Here, $S = (X,\Sigma,\xi,x_0,X_m)$ is a DFA that shares an alphabet with $G$ but has a different state space, transition function and set of marked states~\cite{Ramadge1987}. The second component of the pair, $\phi: X \to \Gamma$, is a state feedback map that maps supervisor states, $x \in X$, to control patterns $\gamma \in \Gamma$,
$$\gamma \coloneqq \phi(x) \in \{0,1\}^\Sigma.$$
These control patterns indicate whether the supervisor will enable an event, denoted by $1$, or disable it, denoted by $0$, with the requirement that no uncontrollable event can be disabled. The controlled DES process is then represented by the supervisor composed with the plant and is denoted by $$\mathcal{S}/G = Ac(X \times Q, \Sigma, \xi \times \delta, (x_0, q_0), X \times Q_m)$$ where $Ac$ is a function that restricts an automaton to its accessible portion: those states that can be reached from the initial state via a string in the automaton's language~\cite{Ramadge1987}. A control pattern can also be defined using the supervisor's transition function, where an event $\sigma$ is enabled in $\phi(x)$ if $\xi(x,\sigma)!$ and it is disabled otherwise. A great amount of work exists in the literature for determining whether a system is controllable, whether a system is observable, and how to synthesize supervisors~\cite{Cassandras2008,Wonham2019}. 

\subsection{The notion of opacity}
When we talk about the opacity property of a system, we use the plain language meaning of an opaque system. That is, we say that a system is opaque if an observer of the system is unable to unambiguously determine some characteristic of that system. Mazar{\'{e}} formalized this notion of opacity for computer science and showed that under this definition it is decidable whether or not a particular system property is opaque~\cite{Mazare2004}. Formulating this notion of opacity for transition systems, Bryans et al.\ used epistemic logic and possible world models to frame opacity as the inability of an observer to determine the truth of a particular predicate related to the system~\cite{Bryans2008}.

Building on this work, Saboori and Hadjicostis characterized a number of state-based opacities for DES including initial state, $k$-step, and infinite step opacities~\cite{Saboori2008,Saboori2011,Saboori2012}. Alternatively, Lin considered opacity to be a broader concept than Bryans et al. and defined opacity in terms of languages, distinguishing between strong and weak opacity in terms of the degree to which strings in one language are confused by an observer with strings in another language~\cite{Lin2011}.

More recently, Wu and Lafortune unified four common kinds of opacity: language-based, initial-state, current-state, and initial-and-final-state. They considered a problem setup where adversaries had full knowledge of the system's structure but only partial observations of its behaviour. Their key result was a set of polynomial-time algorithms for transforming each of these four kinds of opacity into the other under the authors' problem formulation. The one exception to this is that language-based opacity cannot be transformed into initial-state opacity if either the secret or non-secret language is not prefix-closed~\cite{Wu2013}. Without loss of generality, therefore, we consider in this paper the problem of enforcing current-state opacity, which we define in line with Wu and Lafortune~\cite{Wu2013}.
\begin{definition}[Current-state opacity] \label{def:currentStateOpacity}
Given a plant $G = (Q,\Sigma,\delta,q_0,Q_m)$, a projection $P_{\Sigma_o}$, and a disjoint sets of secret states $Q_S \subseteq Q$ and non-secret states $Q_{NS} \subset Q$, $Q_S \cap Q_{NS} = \emptyset$, $G$ is \emph{current-state opaque} if for every string that leads from an initial state to a secret state, there exists a string that leads from an initial state to a non-secret state and these two strings have identical projections.
\begin{align*}
&\forall\ s \in L(G) \text{ such that } \delta(q_0,s) \in Q_S,\\
&\exists\ t \in L(G) \text { such that } \delta(q_0,t) \in Q_{NS} \text{ and } P_{\Sigma_o}(s) = P_{\Sigma_o}(t)
\end{align*}
\end{definition}

\subsection{Enforcing opacity through supervisory control} \label{sec:opacitycontrolproblem}
It is natural to want to ensure that a given system satisfies our well-defined notion of opacity. This problem has been addressed in the literature in many different forms, in this paper we name it the \emph{opacity control problem}. This problem is defined by a set of general parameters~\cite{Bryans2008,Dubreil2010}:
\begin{enumerate}
\item the set of events that are observable and controllable for the supervisor ($\Sigma_s$ and $\Sigma_c$ respectively);
\item the set of events that are observable by the adversary ($\Sigma_a$);
\item the secret and type of opacity to enforce; and
\item the adversary's knowledge of any supervisory control policy.
\end{enumerate}

Since the supervisor and adversary each have a set of observable events, we label these $\Sigma_s$ and $\Sigma_a$ respectively to avoid ambiguity. The most general formulation of the opacity control problem assumes no relationship between $\Sigma_s$, $\Sigma_c$, and $\Sigma_a$. Other reasonable relationships to assume between these three sets is that the adversary sees a subset of the events that the supervisor does, $\Sigma_a \subseteq \Sigma_s$, that the supervisor can only control events that it can see, $\Sigma_c \subseteq \Sigma_s$, and that the system contains only events that are controllable and observable by the supervisor, $\Sigma_c = \Sigma_s = \Sigma$. Regardless of the specific relationship between event sets, the most principled way to approach the opacity control problem is to consider an adversary that is aware of the supervisor's control policy and thereby avoid any appeals to security through obscurity. In this paper we use the formulation given in Problem~\ref{prob:opacitycontrolproblem}.
\begin{problem}[Opacity Control Problem] \label{prob:opacitycontrolproblem}
Given a DFA $G = (Q,\Sigma,\delta,q_0,Q_m)$ where $Q_m$ is the set of secret states, an alphabet of events the supervisor can observe $\Sigma_s \subseteq \Sigma$ and an alphabet of events the supervisor can control $\Sigma_c \subseteq \Sigma_s$. Given an supervisor-aware adversary and an alphabet of events that it can observe $\Sigma_a \subseteq \Sigma_s$. Find a supervisor $\mathcal{S}$ such that the closed-loop behaviour of the plant composed with the supervisor, $\mathcal{S}/G$, is current-state opaque with respect to the alphabet observable by the adversary, $\Sigma_a$, the set of secret states $Q_S = Q_m$, and the set of non-secret states $Q_{NS} = Q \setminus Q_m$.
\end{problem}

The na{\"{i}}ve solution for this formulation is to iteratively calculate the adversary's view of the plant, apply supervisory control, and then check the opacity of the controlled system until a fixed point is reached. This approach has been noted by multiple authors and dismissed as very computationally expensive due to the fact that an observer automaton must be calculated for every iteration~\cite{Dubreil2010,Saboori2012}.

\subsection{Other approaches to the Opacity Control Problem}
The opacity control problem has been addressed in the literature under different assumptions. In some works the adversary is unaware of the supervisor, while in others the adversary has complete knowledge of the supervisor's control policy. We consider the latter case, since this introduces the requirement to iteratively calculate the supervisor's control policy and the adversary's estimate of the plant; Table~\ref{tab:opacityAssumptions} summarizes the characteristics of these works.
\begin{table}[htb]
\centering
\begin{tabular}{c p{8cm} c c}
\toprule
Reference & Approach & Event Sets & Opacity Type \\
\midrule
\cite{Badouel2007} &Sub-lattice of kernels & $\Sigma_a \subseteq \Sigma_s$ & Language-based \\
& & $\Sigma_c = \Sigma$ & \\[0.1cm]
\cite{Takai2008} & Supremal closed, controllable, and opaque sublanguage& $\Sigma_a \subseteq \Sigma_s$ & Language-based \\
& & $\Sigma_c \subseteq \Sigma_s$ & \\
& & $\Sigma_s = \Sigma$ & \\[0.1cm]
\cite{Ben-Kalefa2011} & Supremal controllable, observable, and opaque language& $\Sigma_a = \Sigma_s$ & Language-based \\
& & $\Sigma_c \subseteq \Sigma_s$ & \\[0.1cm]
\cite{Dubreil2010} & Condensed state estimates & $\Sigma_a \subseteq \Sigma_s$ & Current-state \\
& & $\Sigma_c \subseteq \Sigma_s$ & \\[0.1cm]
\cite{Yin2015} & All-Inclusive Controller for Opacity  & $\Sigma_a \subseteq \Sigma_s$ & Current-state  \\[0.1cm]
\cite{Yin2016a} & All-Enforcement Structure & $\Sigma_a = \Sigma_s$ & Current-state \\[0.1cm]
\cite{Tong2018a} & Augmented I-observer & -- & Current-state  \\[0.1cm]
This paper & Plant behaviour composed with the adversary's view & $\Sigma_a \subseteq \Sigma_s$ & Current-state\\
& & $\Sigma_c \subseteq \Sigma_s$ & \\
\bottomrule
\end{tabular}
\caption{Comparing opacity control problem formulations with a supervisor-aware adversary. $\Sigma$ is the set of all events, $\Sigma_c$ is the supervisor's set of controllable events, $\Sigma_s$ is the supervisor's set of observable events, and $\Sigma_a$ is the adversary's set of observable events.}
\label{tab:opacityAssumptions}
\end{table}

Addressing language-based opacity, Badouel et al.\ presented an iterative method to design a supervisor that guaranteed concurrent opacity assuming there are no uncontrollable events in the plant~\cite{Badouel2007}. Takai and Oka proposed a method with one-step convergence under the strong assumption that, for any pair of indistinguishable strings, any uncontrollable and observable event that can occur after one string can occur after the other~\cite{Takai2008}. Finally, Ben-Kalefa and Lin iteratively enforced opacity and controllability for the language generated by the plant~\cite{Ben-Kalefa2011}. This is the na{\"{i}}ve approach applied to languages and still requires the construction of a supervisor to enforce the final language produced.

Considering current-state opacity, Yin and Lafortune presented the \emph{All-Inclusive Controller for Opacity}~\cite{Yin2015} and the \emph{All Enforcement Structure}~\cite{Yin2016a}, which both embed a game between supervisor and plant in a bipartite transition structure where the supervisor's goal is to enforce the specified property and the plant's is to violate the specified property. More similar to our approach to the opacity control problem are \emph{condensed state estimates},~\cite{Dubreil2010}, and the \emph{augmented I-observer},~\cite{Tong2018a}.

Dubreil et al.\ showed that solving the opacity control problem under full observation induces a general solution to the opacity control problem when $\Sigma_a \subseteq \Sigma_s$~\cite{Dubreil2010}. Their approach to enforcing opacity through supervisory control produces an automaton whose states track the plant's actual state as well as condensed state estimates, which contain the set of states that the adversary believes the plant could be in if the last transition it observed was the last transition that occurred in the plant. From the condensed state estimates, the supervisor then reasons about ``loosing paths,'' which are the traces from the current state that would lead to the disclosure of a secret~\cite{Dubreil2010}.

Our approach to the opacity control problem is inspired by the condensed state estimates method, which we improve upon in two ways. First, the parallel composition between the plant and an observer represent the adversary's beliefs about the plant more intuitively than condensed state estimates do. Second, our representation of the adversary's beliefs removes the need for unrolling condensed state estimates to determine whether a particular string will disclose a system secret. 

Finally, Tong et al.'s \emph{augmented I-observer} is the parallel composition between the plant and the observer representing the adversary's view of the plant~\cite{Tong2018a}. The supremal G-opaque sublanguage -- a language-based property equivalent to current-state opacity -- is computed and a supervisor synthesized to enforce it as a specification~\cite[Definition 7]{Tong2018a}. This process iterates until the system's generated language is G-opaque, i.e., the system is current-state opaque. The augmented I-observer method allows $\Sigma_a$ and $\Sigma_s$ to be incomparable, but Tong et al.\ restricted the majority of their discussion to non-supervisor-aware adversaries. With a supervisor-aware adversary, this generality comes with the additional computational cost of updating both the plant and augmented I-observer at every step of the iteration. With the reasonable assumption that $\Sigma_a \subseteq \Sigma_s$, our method is able to simply refine the parallel composition of the plant and adversary's view instead.

\section{Relating Observers: Subautomata versus Subobservers} \label{sec:relatingObservers}
The key to our method is that we are able to relate each successive observer automaton to the one that came before it. To begin, when supervisory control is applied to enforce a regular-language specification for an automaton $G$, the controlled system can be modelled as another automaton, $G'$. We begin, therefore, with an understanding that there is a relationship between $G$ and $G'$. This is the subautomaton relationship, a general relationship that goes beyond supervisory control ~\cite[p.~86]{Cassandras2008}.
\begin{definition}[Subautomaton] \label{def:subautomaton}
We say that $G' = (Q',\Sigma,\delta_{G'},q'_0,Q'_m)$ is a \textbf{subautomaton} of $G = (Q,\Sigma,\delta_G,q_0,Q_m)$, denoted by $G'\ \sqsubseteq\ G$, if $$\delta_{G'}(q'_0,s) = \delta_G(q_{0},s)\quad \forall\ s \in L(G')$$
\end{definition}

The subautomaton relationship is a strong form of correspondence between two automata. It allows us to reason, for example, that the states of $G'$ are a subset of those in $G$, that the initial state in $G'$ is the same as in $G$, and, more generally, that we can match the states in $G'$ with a corresponding state in $G$~\cite[p.~86]{Cassandras2008}. This relationship can be very helpful when reasoning about the language produced by each automaton or about termination conditions for certain algorithms.

When the specification of legal behaviour is a subautomaton of $G$, the resulting automaton $G'$ will also be a subautomaton of $G$. Although the idea of a specification is not directly applicable to the opacity control problem, we demonstrate in Section~\ref{sec:stateSplitting} how our proposed method inherently ensures that supervisory control produces a subautomaton at each iteration.
\begin{figure}[htb]
\centering
\includegraphics[width=0.25\textwidth]{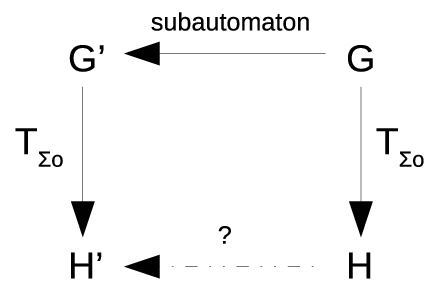}
\caption{We wish to characterize the relationship that exists between $H'$ and $H$.}
\label{fig:motivatingQuestion}
\end{figure}

Beyond the relationship between the uncontrolled and controlled plant, we are interested in how an adversary's observations of a system evolve as supervisory control is enacted. We would like to know whether or not we can establish a relationship between the observer automata $H'$ and $H$, transformations of $G'$ and $G$ respectively with respect to the same alphabet of observable events (Figure \ref{fig:motivatingQuestion}). Does the subautomaton relationship survive the transformation; is $H'$ a subautomaton of $H$? More formally, is it true, given automata $G$, $G'$, $H := T_{\Sigma_o}(G)$ and $H' := T_{\Sigma_o}(G')$, that $$G' \sqsubseteq G \implies H' \sqsubseteq H?$$

Unfortunately this statement is not true because $H'$ and $H$ will have their states coming from $2^{Q'}$ and $2^Q$ respectively, allowing for the possibility that states in $H'$ will not properly match up with states in $H$ (Figure \ref{fig:possibleResultCounterExample}). Although the transition graph of any subautomaton is necessarily a subgraph of the parent automaton's transition graph, the same is not true for automaton related by the subobserver property since states may be ``split'' from $H$ to $H'$. Even if we avoid this, the state labels in $H$ and $H'$ are semantically meaningful, namely they encode the adversary's belief about the state of the system, and these semantics would be lost if we allowed, for example, the state labelled $\{5\}$ to be matched with the state labelled $\{5,6\}$.
\begin{figure}[htb]%
\centering
\includegraphics[width=\textwidth]{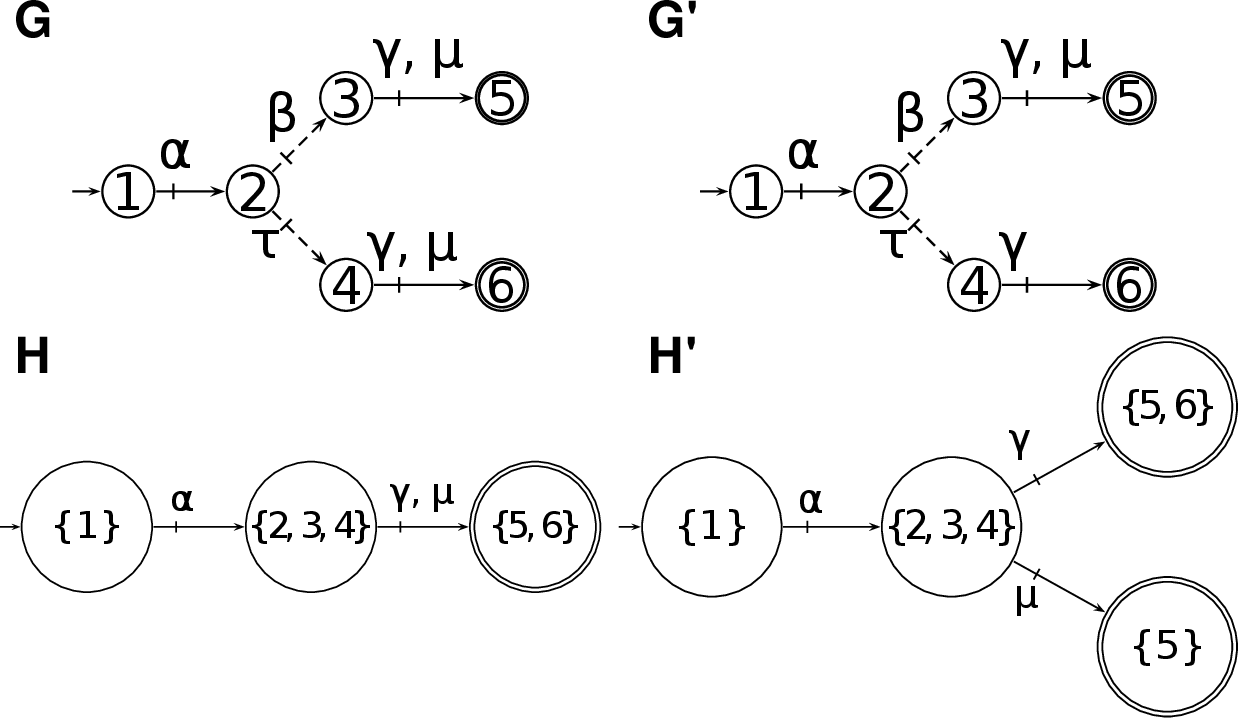}
\caption{Counterexample showing that $H'$ is not a subautomaton of $H$. Initial states are denoted by a small arrow; marked states are denoted with a double ring. Observable events have a solid arrow and unobservable events have a dashed arrow. Controllable events have marks across their arrows.}
\label{fig:possibleResultCounterExample}
\end{figure}

All is not lost, however. Given our insistence that state labels carry semantic meaning, we can relax the definition of subautomaton to allow for this and produce a relationship that is suitable to relate two observers together. We call this the subobserver relationship, Definition~\ref{def:subobserver}.
\begin{definition}[Subobserver] \label{def:subobserver}
We say that $G' = (Q',\Sigma,\delta_{G'},q'_0,Q'_m)$ is a \textbf{subobserver} of $G = (Q,\Sigma,\delta_G,\allowbreak q_0,Q_m)$, denoted $G'\ \tilde{\sqsubseteq}\ G$, if $$\delta_{G'}(q'_0,s) \subseteq \delta_G(q_0,s)\quad \forall\ s \in L(G').$$
\end{definition}
Because we parallel the definition of a subautomaton in our definition of a subobserver it immediately follows that all subautomaton relationships are also subobserver relationships and that the converse is not true. We claim that the subobserver property does complete the relationships between automata $G$, $G'$, $H$, and $H'$ seen in Figure~\ref{fig:motivatingQuestion}.

Before proving this result formally, we begin by establishing some technical notes. To begin, throughout our proofs we will define automata in terms of $G$. Specifically, $G'$ is a subautomaton of $G$ and $H' := T_{\Sigma_o}(G')$ and $H := T_{\Sigma_o}(G)$ are observer automata with our characteristic marking.
\begin{align*}
\begin{aligned}
G =\ &(Q, \Sigma, \delta, q_0, Q_m)\\
G' =\ &(Q', \Sigma, \delta', q_0, Q_m')\\
\end{aligned}
\quad
\begin{aligned}
&H = T_{\Sigma_o}(G) =\ (2^Q, \Sigma_o, \delta_H, q_{0h}, Q_{mh})\\
&H' = T_{\Sigma_o}(G') =\ (2^{Q'}, \Sigma_o, \delta_{H'}, q_{0h'}, Q_{mh'})\\
\end{aligned}
\end{align*}
Next, we define a state's $\epsilon$-reach as the set of states that can be reached from it via the empty string (Definition~\ref{def:epsilonreach})~\cite[p.~71]{Cassandras2008}.
\begin{definition}[$\epsilon$-reach] \label{def:epsilonreach}
For an automaton $G$ and projection $P_{\Sigma_o}$, we call the set of states that can be reached from state $q$ via a string whose projection is $\epsilon$ as the \emph{$\epsilon$-reach} of $q$. Formally, $$\epsilon R_G(q) \coloneqq \{ q' \in Q \ |\ \exists\ s \in \Sigma^*, P_{\Sigma_o}(s) = \epsilon, \delta(q,s) = q'\}.$$ By definition a state is always in its own $\epsilon$-reach, $q \in \epsilon$R$(q)$.
\end{definition}
We also show that if $G' \sqsubseteq G$ then a state's $\epsilon$-reach in $G'$ is a subset of that state's $\epsilon$-reach in $G$.
\begin{lemma}[The $\epsilon$-reach of any state in a subautomaton is a subset of that state's $\epsilon$-reach in the parent automaton] \label{lem:epsilonReach}
Given automata $G = (Q,\Sigma,\delta_G,q_0,Q_m)$ and $G' = (Q',\Sigma,\delta_{G'},q_{0'},Q'_m)$. Then $$G' \sqsubseteq G \implies \epsilon R_{G'}(q) \subseteq \epsilon R_{G}(q)\ \forall\ q \in Q'.$$
\end{lemma}
\begin{proof}
To begin, $G' \sqsubseteq G$ implies that $Q' \subseteq Q$. We can therefore determine the $\epsilon$-reach of state $q$ with respect to both $G$ and $G'$ for all $q \in Q'$.
\begin{align*}
\epsilon R_{G'}(q) =&\{ q' \in Q'\ |\ \exists\ s \in \Sigma^*, P_{\Sigma}(s) = \epsilon, \delta_{G'}(q,s) = q'\}\\
\subseteq &\{ q' \in Q' \ |\ \exists\ s \in \Sigma^*, P_{\Sigma}(s) = \epsilon, \delta_{G'}(q,s) = q'\}\\
&\cup \{q' \in Q'\ |\ \exists\ s \in \Sigma^*,P_{\Sigma}(s) = \epsilon,\delta_{G}(q,s) = q', \neg \delta_{G'}(q,s)!\}\\
&\cup \{ q' \in Q \setminus Q' \ |\ \exists\ s \in \Sigma^*, P_{\Sigma}(s) = \epsilon, \delta_G(q,s) = q'\}\ \tag*{by set inclusion}\\
= &\{ q' \in Q\ |\ \exists\ s \in \Sigma^*, P_{\Sigma}(s) = \epsilon, \delta_G(q,s) = q'\}\\
= & \epsilon R_{G}(q)
\end{align*}
\qed
\end{proof}
We are now ready to prove that the relationship between the observer automata $H$ and $H'$ (the dashed line in Figure~\ref{fig:motivatingQuestion}) is the subobserver relationship defined in Definition~\ref{def:subobserver}.
\begin{theorem}[The transformation of a subautomaton is itself a subobserver] \label{lem:realResult}
Given automata $G$ and $G'$ and their respective transformations, $H := T_{\Sigma_o}(G)$ and $H':=T_{\Sigma_o}(G')$. Then $$G' \sqsubseteq G \implies H'\ \tilde{\sqsubseteq}\ H.$$
\end{theorem}
\begin{proof}
The proof proceeds by induction.

\textbf{Base case}: $s = \epsilon$, $|s| = 0$.
\begin{align*}
q_{0h'} &= \epsilon R_{G'}(q_0)\\
&\subseteq \epsilon R_{G}(q_0) \tag*{by Lemma 1}\\
&= q_{oh}
\end{align*}

\textbf{Inductive Hypothesis}: Suppose $s \in L(H')$, $|s| = n$, and $\bar{Q'} = \delta_{H'}(q_{0h'},s) \subseteq \delta_H( q_{0h},s) = \bar{Q}$.

\textbf{Inductive Step}: Now we show that for $s\sigma \in L(H')$, $|s\sigma| = n+1$, then $\delta_{H'}(q_{0h'},s\sigma) \subseteq \delta_H(q_{0h},s\sigma)$.

We proceed by translating the set of states $\delta_{H'}(q_{0h'},s\sigma)$ into set builder notation and then adding well-chosen states to produce the set of states $\delta_H(q_{0h},s\sigma)$.
\begin{align*}
\delta_{H'}(q_{0h'},s\sigma) = &\delta_{H'}(\bar{Q'},\sigma) \tag*{by definition in our IH}\\
= &\bigcup_{\substack{q' \in \bar{Q'}\\\delta_{G'}(q',\sigma)!}} \epsilon R_{G'}(\delta_{G'}(q',\sigma))\tag*{by definition of $H'$}\\
\subseteq &\bigcup_{\substack{q' \in \bar{Q'}\\\delta_{G'}(q',\sigma)!}} \epsilon R_{G}(\delta_{G'}(q',\sigma)) \tag*{by Lemma 1}\\
\subseteq &\bigcup_{\substack{q' \in \bar{Q'}\\\delta_{G'}(q',\sigma)!}} \epsilon R_{G}(\delta_{G'}(q',\sigma))\quad \cup  \bigcup_{\substack{q' \in \bar{Q'}\\\neg \delta_{G'}(q',\sigma)!, \delta_{G}(q',\sigma)!}} \epsilon R_{G}(\delta_G(q',\sigma))\tag*{by set inclusion}\\
\subseteq &\bigcup_{\substack{q' \in \bar{Q'}\\\delta_{G'}(q',\sigma)!}} \epsilon R_{G}(\delta_{G'}(q',\sigma))\quad \cup  \bigcup_{\substack{q' \in \bar{Q'}\\\neg \delta_{G'}(q',\sigma)!, \delta_{G}(q',\sigma)!}} \epsilon R_{G}(\delta_G(q',\sigma))\quad \cup\\
&\bigcup_{\substack{q \in \bar{Q} \setminus \bar{Q'}\\\delta_{G}(q,\sigma)!}} \epsilon R_{G}(\delta_G(q,\sigma))\tag*{by set inclusion}\\
= &\bigcup_{\substack{q' \in \bar{Q'}\\\delta_{G}(q',\sigma)!}} \epsilon R_{G}(\delta_G(q',\sigma))\quad \cup \bigcup_{\substack{q \in \bar{Q} \setminus \bar{Q'}\\\delta_{G}(q,\sigma)!}} \epsilon R_{G}(\delta_G(q,\sigma))\tag*{combining the first two terms}\\
=  &\bigcup_{\substack{q \in \bar{Q}\\\delta_{G}(q,\sigma)!}} \epsilon R_{G}(\delta_G(q,\sigma))\tag*{since by IH $\bar{Q'} \subseteq \bar{Q}$}\\
= &\delta_{H}(\bar{Q},\sigma)\\
= &\delta_{H}(q_{0h},s\sigma)\tag*{by definition in our IH}
\end{align*}
\qed
\end{proof}
Having proven Theorem \ref{lem:realResult}, we can relate our four automata from Figure~\ref{fig:motivatingQuestion}: $G'$ is a subautomaton of $G$, $H'$ and $H$ are transformations of $G'$ and $G$ respectively, and $H'$ is a subobserver of $H$.

\section{Updating the Observer Automaton} \label{sec:updatingObservers}
An adversary's view of a plant is tied to the plant's structure, so we explicitly link the adversary's plant estimate and the plant's actual state through the parallel composition of their respective automata (Definition~\ref{def:parallelComposition}~\cite[p.~80]{Cassandras2008}).
\begin{definition}[Parallel Composition] \label{def:parallelComposition}
Given two automata $G_1 = (Q_1, \Sigma_1, \delta_1, q_{01}, Q_{m1})$ and $G_2 = (Q_2, \Sigma_2, \delta_2, q_{02}, Q_{m2})$, their \textbf{parallel composition} is defined as:
\begin{align*}
G_1\ ||\ G_2 = &Ac(Q_1 \times Q_2, \Sigma_1 \cup \Sigma_2, \delta_{12}, (q_{01},q_{02}), Q_{m1} \times Q_{m2})\\
&\text{where } \delta_{12}((q_1,q_2),\sigma)= \begin{cases}(\delta_1(q_1,\sigma),\delta_2(q_2,\sigma))\quad\text{if }\sigma \in \Sigma_1 \cap \Sigma_2\\
(\delta_1(q_1,\sigma),q_2)\quad\text{if }\sigma \in \Sigma_1 \setminus \Sigma_2\\
(q_1,\delta_2(q_2,\sigma))\quad\text{if }\sigma \in \Sigma_2 \setminus \Sigma_1\\
\text{undefined}\quad\text{otherwise}\\
\end{cases}
\end{align*}
\end{definition}

The parallel composition $G\ ||\ H$, contains all of the information required to track the evolution of the plant's behaviour and the adversary's beliefs about the plant. We will denote states in $G\ ||\ H$ using the notation $\frac{q}{A}$ where $q$ is a state in $G$ and $A$ is a state in $H$. It is useful to observe that a state $\frac{q}{A}$ is accessible in $G\ ||\ H$ if and only if there is a string $s \in L(G)$ that leads to state $q$ in $G$ and whose projection leads to estimate $A$ in $H$ (Lemma~\ref{lem:stateInGH}). As a corollary to this Lemma, if a state $\frac{q}{A}$ is accessible in $G\ ||\ H$ then the state $q$ is an element of the estimate $A$.
\begin{lemma} \label{lem:stateInGH}
Given the parallel composition $G\ ||\ H$ where $G = (Q,\Sigma,\delta_G,q_0,Q_m)$ is a DFA and $H := T_{\Sigma_o}(G)$ is an observer automaton. A state $\frac{q}{A}$ is accessible in $G\ ||\ H$ if and only if $$(\exists s \in L(G))(\delta_G(q_0,s) = q)(\{q' \in Q\ |\ (\exists s' \in L(G)\ |\ P_{\Sigma_o}(s) = P_{\Sigma_o}(s'))(\delta_G(q_0,s')=q')\} = A).$$
\end{lemma}
\begin{proof}
See the Appendix.
\end{proof}

The parallel composition has strong properties; first, it follows that $L(G\ ||\ H) = L(G)$~\cite{Jiraskova2012}. More interesting is that the accessible states of $T_{\Sigma_o}(G\ ||\ H)$ are pairwise disjoint sets of states of $G\ ||\ H$, i.e., every state in $G\ ||\ H$ belongs to one and only one state in $T_{\Sigma_o}(G\ ||\ H)$. An automaton that satisfies this property is a \emph{state-partition automaton} (SPA). The SPA property allows for simple refinement operations, but we will not explicitly use it as Cho and Marcus used it,~\cite{Cho1989a}, since exploiting the SPA property would require a second transformation operation, $T_{\Sigma_o}(G\ ||\ H)$, and would introduce additional preprocessing overhead to our algorithms.

\subsection{Using the subobserver property}
We are given the parallel composition of a plant, $G$, and an adversary's view of that plant, $H$. Given that the plant's behaviour has been restricted, the \texttt{REFINE} algorithm (Algorithm~\ref{alg:implicitCalculations}) returns the parallel composition of the restricted plant and the adversary's updated view of the plant without explicitly calculating this adversary view. To begin, we remove states from the parallel composition that should be made inaccessible. Next, we account for the adversary's ability to reason about the restricted plant. If there is a string $s$ in $L(G)$ that is no longer possible in $G'$, then the adversary cannot confuse other strings in $L(G')$ with $s$. We must therefore remove the state that $s$ leads to in $G$ from estimates where the adversary confused $s$ with other strings.

To simplify our reasoning, we assume that the plant $G$ and the parallel composition $G\ ||\ H$ have isomorphic state transition diagrams. Although this is not always the case, we will demonstrate in Section~\ref{sec:stateSplitting} that the parallel composition operation inherently creates an automaton $\hat{G}$ that is language equivalent to $G$ and whose state transition diagram is isomorphic with that of $G\ ||\ H$.
\begin{algorithm}[htb]
\KwData{Automaton $G\ ||\ H$, $\Delta$, a list of states to make inaccessible in $G\ ||\ H$.}
\KwResult{An automaton $M$, the accessible portion of $G'\ ||\ H'$.}
$M = G\ ||\ H$\; \label{alg:denoteByM}
\tcc{$A_r(A)$ is a set of states that should be removed from the adversary's estimate $A$.}
Initialize $A_r(\cdot)$ as $\emptyset$ for each estimate\; \label{alg:initializeAr}
\tcc{The set $\Delta_Q$ compiles all of the states that must be made inaccessible to make the states in $\Delta$ inaccessible. We therefore start by adding all of the states in $\Delta$ to the set $\Delta_Q$.}
$\Delta_Q = \Delta$ \label{alg:deltaToDeltaQ}\;
\While{$\Delta_Q$ is not empty \label{alg:buildDeltaQStart}}
{Take $\frac{q}{A}$ as an element of $\Delta_Q$\; 
\ForEach{transition leading into $\frac{q}{A}$, $t = \langle \frac{q'}{A'}, \sigma, \frac{q}{A} \rangle_M$}
{
\uIf(the transition $t$ is uncontrollable){$\sigma \in \Sigma_{uc}$}
{Add $\frac{q'}{A'}$ to $\Delta_Q$ \label{alg:uncontrollableToDeltaQ}\;}
}
Remove $\frac{q}{A}$ from $\Delta_Q$\;
Remove $\frac{q}{A}$ from $M$ \label{alg:removeStateInDeltaQ}\;
}
\tcc{States in $M$ sharing the same estimate as newly inaccessible states need to be updated.}
\ForEach{State $\frac{q}{A}$ in $G\ ||\ H$ \label{alg:checkInaccessibleStatesStart}}
{
\If{$\frac{q}{A}$ is not accessible in $M$ \label{alg:testAccessToBuildAr}}
{
$A_r(A) = A_r(A) \cup \{q\}$\; \label{alg:buildAr}
}
} \label{alg:checkInaccessibleStatesEnd}
\tcc{Relabel states in $M$ to account for the observer's updated estimates of the plant.}
\ForEach{State $\frac{q}{A}$ in $M$ \label{alg:accessibleStart}}
{
{Relabel $\frac{q}{A}$ to $\frac{q}{A \setminus A_r(A)}$\;}
} \label{alg:accessibleEnd}
\tcc{Ensure only the accessible portion of $M$ is returned}
$M = Ac(M)$\; \label{alg:onlyAccessible}
\caption{The \texttt{REFINE} algorithm}
\label{alg:implicitCalculations}
\end{algorithm}

\subsection{Technical developments}
Before proving that the \texttt{REFINE} algorithm is correct, we need to develop a few ideas that will be necessary to connect our state-based formulation of opacity with the string-based confusion of the adversary. First, we know that \texttt{REFINE} is guaranteed to produce a subautomaton of $G$ in the first elements of the states of $M$ because it removes states from $G\ ||\ H$.
\begin{lemma} \label{lem:produceGPrimeSubautomaton}
Given the parallel composition $G\ ||\ H$, where $G = (Q,\Sigma,\delta_G,q_0,Q_m)$ is a DFA and $H := T_{\Sigma_o}(G)$ is an observer automaton, and $\Delta$ is a list of states to remove from $G\ ||\ H$. Then the first elements of the states remaining in $G\ ||\ H$ represent a subautomaton of $G$.
\end{lemma}
\begin{proof}
See the Appendix.
\end{proof}
Second, we note that a subautomaton $G'$ can be defined by its parent automaton $G$ and the set of strings that have been removed from $L(G)$. We define an equivalence relation $\rho$ on $\Sigma^*$, which relates strings if they are both in $L(G')$, were both removed from $L(G)$, or were never in $L(G)$.
\begin{definition} \label{def:rhoRelation}
The binary relation $\rho_{G,G'}$ on $\Sigma^*$ is defined as $(s,s') \in \rho_{G,G'}$ if and only if $$(s, s' \in L(G'))\ \vee\ (s, s' \in L(G) \setminus L(G'))\ \vee\ (s, s' \in \Sigma^* \setminus L(G)).$$
\end{definition}
Third, in \texttt{REFINE} we treat strings the same if they lead to the same state in $G$ and the same estimate in $H$. We define an equivalence relation $match_{\Sigma_o}$ on $\Sigma^*$ to capture this, where strings are related if they will receive identical treatment from \texttt{REFINE}.
\begin{definition} \label{def:matchRelation}
Given a DFA $G$ and observer automaton $H \coloneqq T_{\Sigma_o}(G)$. The binary relation $match_{\Sigma_o}$ on $\Sigma^*$ is defined as $$(s,s') \in match_{\Sigma_o} \iff (s \equiv s' \mod eq(G)) \wedge (P_{\Sigma_o}(s) \equiv P_{\Sigma_o}(s') \mod eq(H))$$ and if $s \notin L(G)$ then $(s,s') \in \mu$ if and only if $s' \notin L(G)$ as well.
\end{definition}
We can show that the partition defined by $match_{\Sigma_o}$ is equivalent to the partition defined by $eq(G\ ||\ H)$.
\begin{lemma} \label{lem:matchAndEquiresponse}
Given a DFA $G = (Q,\Sigma,\delta_G,q_0,Q_m)$, a set of observable events $\Sigma_o \subset \Sigma$, and an observer automaton $H \coloneqq T_{\Sigma_o}(G)$. Then $$match_{\Sigma_o} = eq(G\ ||\ H).$$
\end{lemma}
\begin{proof}
We will show that two strings $s, s' \in L(G)$ are related in $match_{\Sigma_o}$ if and only if they are related in $eq(G\ ||\ H)$.
\begin{align*}
&(s,s') \in match_{\Sigma_o}\\
 &\iff (s \equiv s' \mod eq(G)) \wedge (P_{\Sigma_o}(s) \equiv P_{\Sigma_o}(s') \mod eq(H)) \tag*{by Definition~\ref{def:matchRelation}}\\
&\iff  \delta_G(q_0,s) = \delta_G(q_0, s') \wedge \delta_H(A_0,P_{\Sigma_o}(s)) = \delta_H(A_0,P_{\Sigma_o}(s')) \tag*{by Definition~\ref{def:equiresponseRelation}}\\
&\iff \delta_{G\ ||\ H}(\frac{q_0}{A_0},s) =  \delta_{G\ ||\ H}(\frac{q_0}{A_0},s') \tag*{by Definition~\ref{def:parallelComposition}}\\
&\iff s \equiv s' \mod eq(G\ ||\ H)\tag*{by Definition~\ref{def:equiresponseRelation}}
\end{align*}
\qed
\end{proof}
Finally, we can build on Lemma~\ref{lem:matchAndEquiresponse} by showing that if $match_{\Sigma_o}$ refines $\rho_{G,G'}$, then we can produce $H'$ by relabelling states in $H$ (i.e., the state transition graph of $H'$ is guaranteed be a subgraph of $H$'s state transition graph). This condition implicitly enforces that \texttt{REFINE} must remove states from $G\ ||\ H$ and thereby treat strings that lead to the same state and estimate identically. This is a reasonable restriction since we are considering state-based formulations of opacity and because alternate automaton representations of a language can be used if strings currently related by $match_{\Sigma_o}$ should be treated differently.
\begin{lemma} \label{lem:hPrimeByRelabelling}
Given DFAs $G = (Q,\Sigma,\delta_G,q_0,Q_m)$, a set of observable events $\Sigma_o \subset \Sigma$, and an observer automaton $H = T_{\Sigma_o}(G)$. If $G'$ is a subautomaton of $G$ and $H' = T_{\Sigma_o}(G')$, then $$(\forall s,s' \in L(H))(s \equiv s' \mod eq(H))\ \wedge\ match_{\Sigma_o} \leq \rho_{G,G'} \implies s \equiv s' \mod eq(H').$$
\end{lemma}
\begin{proof}
The proof proceeds by contradiction. Assume that $(match_{\Sigma_o} \leq \rho_{G,G'})$ and that there exists strings $s$ and $s'$ in $L(H)$ such that $s \equiv s' \mod eq(H)$ but $s \not\equiv s' \mod eq(H')$.
\begin{align} \label{eq:sAndSPrimeCongruent}
s \equiv s' \mod eq(H) \implies &\delta_H(A_0,s) = \delta_H(A_0,s')
\end{align}
We can enumerate exactly those states that appear in $\delta_H(A_0,s)$, and therefore that appear in $\delta_H(A_0,s')$ as well.
\begin{align*}
\delta_H(A_0,s) = &\{q \in Q\ |\ (\exists t \in L(G))(P_{\Sigma_o}(t) = s)(\delta_G(q_0,t) = q)\}\\
= & \{q \in Q\ |\ (\exists t \in L(G'))(P_{\Sigma_o}(t) = s)(\delta_G(q_0,t) = q)\} \cup\\
& \{q \in Q\ |\ (\exists t \in L(G) \setminus L(G'))(P_{\Sigma_o}(t) = s)(\delta_G(q_0,t) = q)\}\\
= &\delta_{H'}(A_0,s) \cup \{q \in Q\ |\ (\exists t \in L(G) \setminus L(G'))(P_{\Sigma_o}(t) = s)(\delta_G(q_0,t) = q)\}
\end{align*}
For notational convenience, we will denote the second term as
\begin{align} \label{eq:QR}
Q_\rho(s) \coloneqq \{q \in Q\ |\ (\exists t \in L(G) \setminus L(G'))(P_{\Sigma_o}(t) = s)(\delta_G(q_0,t) = q)\}.
\end{align}
These sets of states are disjoint, which can be shown by contradiction. Assume that there exists a state $q \in Q$ that belongs to the intersection of these two sets.
\begin{align*}
q \in \delta_{H'}(A_0,s) \cap Q_\rho(s) \implies &(\exists t \in L(G'), t' \in L(G) \setminus L(G'))\\
&(\delta_G(q_0,t)=\delta_G(q_0,t')=q)\\
&(P_{\Sigma_o}(t) = P_{\Sigma_o}(t') = s)\tag*{definition of $\delta_{H'}(A_0,s)$ and $Q_\rho(s)$}\\
\implies & (\exists t \in L(G'), t' \in L(G) \setminus L(G'))\\
&(\delta_G(q_0,t)=\delta_G(q_0,t')=q)\\
&(\delta_H(A_0,P_{\Sigma_o}(t)) = \delta_H(A_0,P_{\Sigma_o}(t'))) \tag*{since $P_{\Sigma_o}(t) = P_{\Sigma_o}(t')$}\\
\implies &(\exists t \in L(G'), t' \in L(G) \setminus L(G'))\\
&(t \equiv t' \mod eq(G))\\
&(P_{\Sigma_o}(t) \equiv P_{\Sigma_o}(t') \mod eq(H)) \tag*{by Definition~\ref{def:equiresponseRelation}}\\
\implies &(\exists t \in L(G'), t' \in L(G) \setminus L(G'))\\
&((t,t') \in match_{\Sigma_o}) \tag*{by Definition~\ref{def:matchRelation}}\\
\implies &(\exists t \in L(G'), t' \in L(G) \setminus L(G'))\\
&(t,t' \in L(G') \vee t,t' \in L(G) \setminus L(G')) \tag*{since $match_{\Sigma_o} \leq \rho_{G,G'}$}
\end{align*}
Since $L(G')$ and $L(G) \setminus L(G')$ are disjoint, this is a contradiction. We can therefore say that $\delta_{H'}(A_0,s) \cap Q_\rho(s) = \emptyset$ and that
\begin{align} \label{eq:removeQR}
\delta_{H'}(A_0,s) = \delta_{H}(A_0,s) \setminus Q_\rho(s).
\end{align}
For the sake of contradiction, we assumed that $s \not\equiv s' \mod eq(H')$. By Definition~\ref{def:equiresponseRelation} this implies that $\delta_{H'}(A_0,s) \neq \delta_{H'}(A_0,s')$. From \eqref{eq:removeQR} we know for $s$ and $s'$ that
\begin{align*}
\delta_{H'}(A_0,s) &= \delta_{H}(A_0,s) \setminus Q_\rho(s)\\
\delta_{H'}(A_0,s') &= \delta_{H}(A_0,s') \setminus Q_\rho(s')
\end{align*}
and since we have reasoned that $\delta_{H}(A_0,s) =\delta_{H}(A_0,s')$, \eqref{eq:sAndSPrimeCongruent}, this implies that $Q_\rho(s) \neq Q_\rho(s')$.
\begin{align*}
Q_\rho(s) \neq Q_\rho(s') \implies &(\exists q^\star \in Q)(q^\star \in Q_\rho(s))(q^\star \not\in Q_\rho(s')) \tag*{WLOG, since $eq(H)$ is symmetric}\\
\implies &(\exists t \in L(G) \setminus L(G'))(P_{\Sigma_o}(t) = s)(\delta_G(q_0,t) = q^\star)\\
&(\not\exists t' \in L(G) \setminus L(G'))(P_{\Sigma_o}(t') = s')(\delta_G(q_0,t') = q^\star) \tag*{by \eqref{eq:QR}}\\
\implies &(\exists t \in L(G) \setminus L(G'))(P_{\Sigma_o}(t) = s)(\delta_G(q_0,t) = q^\star)\\
&(\exists t' \in L(G'))(P_{\Sigma_o}(t') = s')(\delta_G(q_0,t') = q^\star) \tag*{since $L(G') \cap L(G) \setminus L(G') = \emptyset$}\\
\implies &(\exists t \in L(G) \setminus L(G'), t' \in L(G'))(P_{\Sigma_o}(t) = s)(P_{\Sigma_o}(t') = s')\\
&(t \equiv t' \mod eq(G)) \tag*{by Definition~\ref{def:equiresponseRelation}}\\
\implies &(\exists t \in L(G) \setminus L(G'), t' \in L(G'))\\
&(t \equiv t' \mod eq(G\ ||\ H)) \tag*{since $s \equiv s' \mod eq(H)$}\\
\implies &(\exists t \in L(G) \setminus L(G'), t' \in L(G'))\\
&((t,t') \in match_{\Sigma_o} \tag*{by Lemma~\ref{lem:matchAndEquiresponse}}\\
\implies &(\exists t \in L(G) \setminus L(G'), t' \in L(G'))\\
&(t,t' \in L(G') \vee t,t' \in L(G) \setminus L(G')) \tag*{since $match_{\Sigma_o} \leq \rho_{G,G'}$}
\end{align*}
Again, since we know that $L(G')$ and $L(G) \setminus L(G')$ are disjoint, this is a contradiction. This implies that $Q_\rho(s) = Q_\rho(s')$, which further implies that $s \equiv s' \mod eq(H')$. We have therefore proved that $$(\forall s,s' \in L(H))(s \equiv s' \mod eq(H))\ \wedge\ match_{\Sigma_o} \leq \rho_{G,G'} \implies s \equiv s' \mod eq(H').$$
\qed
\end{proof}
Lemma~\ref{lem:hPrimeByRelabelling} guarantees us that for any state in $H$---cell in the partition defined by $eq(H)$---we can produce the corresponding state in $H'$ by relabelling that state. It goes beyond this too, by telling us that this relabelling removes exactly those states in $Q_\rho$ from the state label in $H$.

\subsection{Proof of correctness for the \texttt{REFINE} algorithm}
As a first point, the \texttt{REFINE} algorithm is guaranteed to terminate because it works by removing and relabelling states in the input DFA, $G\ ||\ H$. For the first loop, the number of states and transitions to be inspected is finite and each state and transition can only be inspected once, therefore the loop terminates. For the second loop, there are a finite number of states to assess for accessibility and no changes are made to the structure of $M$. For the third loop, there are a finite number of states to relabel and relabelling a states in $M$ never requires relabelling another state in $M$, therefore the loop terminates. Finally, taking the accessible portion of $M$ requires assessing the accessibility of a finite number of states, so this operation also terminates.

Secondly, we claim that the \texttt{REFINE} algorithm uses the subobserver relationship to calculate $G'\ ||\ H'$ without computing $H'$ explicitly. This is formalized as Theorem~\ref{the:correctnessOfRefine}.
\begin{theorem} \label{the:correctnessOfRefine}
Given the parallel composition $G\ ||\ H$ and list of states $\Delta$ where: $G = (Q,\Sigma,\delta_G,q_0,Q_m) $ is a DFA; $H := T_{\Sigma_o}(G)$ is an observer automaton; and $\Delta$ is a list of states in $G\ ||\ H$. Denote by $G'$ the automaton that results in the first elements of $G\ ||\ H$ once the states in $\Delta$ are made inaccessible. Then the \texttt{REFINE} algorithm produces $G'\ ||\ H'$ where $H' := T_{\Sigma_o}(G')$.
\end{theorem}
\begin{proof}
We show that, given $G\ ||\ H$, where $H$ is the transformed automaton $T_{\Sigma_o}(G)$, and $\Delta$, a list of states to remove from $G\ ||\ H$ that produces $G'$ in the first elements of the states in $G$, the \texttt{REFINE} algorithm constructs $G'\ ||\ H'$.

We assume that $\Delta$ is not equal to the set of states in $G\ ||\ H$, otherwise $G'$ is the null automaton and we simply need to return $G'\ ||\ H' = null$. For the case where $\Delta$ is not equal to the set of states in $G\ ||\ H$, we reason about the elements of the states separately, since these correspond to the states of $G'$ and $H'$.

First, we argue that the first loop removes the correct states from $M$ to make the states in $\Delta$ inaccessible. For the purposes of this proof, we define a new set $\Delta_Q^{All}$ as the set of all states that were added to $\Delta_Q$. If we wanted to include this set in our implementation of \texttt{REFINE}, we could an operation after Line~\ref{alg:deltaToDeltaQ} and after Line~\ref{alg:uncontrollableToDeltaQ} where the state that was added to $\Delta_Q$ is also added to $\Delta_Q^{All}$. Since states are never removed from $\Delta_Q^{All}$, it will contain all of the states that were at one point in $\Delta_Q$ and were removed from the automaton $M$. This allows us to reason that after the first loop:
\begin{equation*}
\frac{q}{A} \in \Delta_Q^{All} \iff (\frac{q}{A} \in \Delta)\ \vee (\frac{q}{A}\text{ leads uncontrollably to a state in }\Delta_Q^{All}.)
\end{equation*} 
\begin{enumerate}
\item $\implies$. States are in $\Delta_Q^{All}$ for only two reasons. The first is if they are added to $\Delta_Q$ at Line~\ref{alg:deltaToDeltaQ}, in which case they are in $\Delta$. The second is if they are added to $\Delta_Q$ at Line~\ref{alg:uncontrollableToDeltaQ}, in which case they lead uncontrollably to a state in $\Delta_Q^{All}$. So every state in $\Delta_Q^{All}$ meets one of the two conditions in the consequent.
\item $\impliedby$. First, Line~\ref{alg:deltaToDeltaQ} makes $\Delta_Q$ equal to $\Delta$ before beginning the loop, which means that all states in $\Delta$ are also in $\Delta_Q^{All}$. Second, for any state $\frac{q}{A}$ that leads uncontrollably to a state $\frac{q'}{A'}$ in $\Delta_Q^{All}$, there is an uncontrollable transition that leads from $\frac{q}{A}$ to $\frac{q'}{A'}$. When the iteration for the state $\frac{q'}{A'}$ is reached in the loop, that transition will be identified as uncontrollable and $\frac{q}{A}$ will be added to $\Delta_Q$ at Line~\ref{alg:uncontrollableToDeltaQ} and therefore will be added to $\Delta_Q^{All}$ as well.

Since states are never removed from $\Delta_Q^{All}$, these two lines of reasoning guarantee that all states in $\Delta$ are in $\Delta_Q^{All}$ and that all states that lead uncontrollably to a state in $\Delta_Q$ are in $\Delta_Q^{All}$ as well.
\end{enumerate}
After the first loop is complete, $\Delta_Q^{All}$ contains exactly those states that must be removed to make the states in $\Delta$ inaccessible. These are exactly the states that are removed from $M$, since any state that appears in $\Delta_Q^{All}$ was in $\Delta_Q$ and was therefore removed from $M$ at Line~\ref{alg:removeStateInDeltaQ}. \texttt{REFINE} doesn't remove any other states from $M$. Therefore the only states that have been removed from $M$ are those that were necessary to make the set of states in $\Delta$ inaccessible. We denote the first element of the remaining states in $M$ as an automaton $G'$.

Second, because we created $G'$ by removing states from $M = G\ ||\ H$, we know by Lemma~\ref{lem:produceGPrimeSubautomaton} that $G' \sqsubseteq G$. Since $G' \sqsubseteq G$, $H \coloneqq T_{\Sigma_o}(G)$, and $H' \coloneqq T_{\Sigma_o}(G')$, we know from Theorem 1 that $H' \tilde{\sqsubseteq} H$. We argue that because $H' \tilde{\sqsubseteq} H$, the estimates related to some strings need to be updated and that this can be accomplished by relabelling the states. We also argue that because the strings removed from $L(G)$ to produce $L(G')$ were determined by a set of states in $G\ ||\ H$, we have that $match_{\Sigma_o} \leq \rho_{G,G'}$ by Lemma~\ref{lem:matchAndEquiresponse} and therefore relabelling the estimates from $H$ is all that is required to produce $H'$ (Lemma~\ref{lem:hPrimeByRelabelling}).

The second loop in \texttt{REFINE} produces a list of states $A_r(A)$ for each estimate $A$. We claim that a state $q$ will be added to $A_r(A)$ if and only if $q$ was in the estimate $A$ to begin with and there are no longer any strings in $L(G')$ that lead to $q$ and which share projections with strings leading to the other states in $A$:
$$q \in A_r(A) \iff (q \in A) \wedge ((\forall s \in L(G)\ |\ \delta_G(q_0,s) = q \wedge \delta_H(A_0,P_{\Sigma_o}(s)) = A)(s \notin L(G'))).$$
We will reason about the two directions separately.
\begin{enumerate}
\item $\implies$. If $q \in A_r(A)$ then it was added at Line~13 because the state $\frac{q}{A}$ was in $G\ ||\ H$ but no longer accessible in $M$ at Line~11. Because the state $\frac{q}{A}$was accessible in $G\ ||\ H$, Lemma~2 tells us that $q \in A$, satisfying the consequent's first expression.

For the consequent's second expression, we will proceed by contradiction. Assume for the sake of contradiction that there exists a string $s \in L(G)$ such that $$\delta_G(q_0,s) = q \wedge \delta_H(A_0,P_{\Sigma_o}(s)) = A$$ and that this string $s$ is also in $L(G')$. Since $s$ is in $L(G')$ and since the first elements of states remaining in $M$ after the first loop exactly reflect $G'$, this means that the state to which $s$ leads has not been removed from $M$ and hence $\frac{q}{A}$ remains accessible in $M$ at Line~12. We therefore know that $q$ would not have been added to $A_r(A)$ at Line~13, which is a contradiction. Therefore $q \in A_r(A)$ implies the consequent's second expression as well.

\item $\impliedby$. Because $(q \in A) \wedge (\forall s \in L(G)\ |\ \delta_G(q_0,s) = q \wedge \delta_H(A_0,P_{\Sigma_o}(s)) = A)$ we have that the state $\frac{q}{A}$ was accessible in $G\ ||\ H$. The last expression in the antecedent tells us, however, that any string $s$ that led to $\frac{q}{A}$ in $G\ ||\ H$ is not in $L(G')$. Since the first elements of states remaining in $M$ after the first loop exactly reflect $G'$, this means that $\frac{q}{A}$ is no longer accessible in $M$ at Line~12. The state $q$ is therefore added to $A_r(A)$ at Line~13. Since states are never removed from $A_r(A)$, this guarantees that $q \in A_r(A)$.
\end{enumerate}

Third, the third loop correctly updates all of the estimates in $M$ because all states in $M$ are visited and because all states in $M$ sharing an estimate must be updated in the same way (Lemma~\ref{lem:hPrimeByRelabelling}).

To conclude, after the \texttt{REFINE} algorithm has completed its work, the resulting automaton $M$ is equal to $G'\ ||\ H'$ such that $$\delta_M(\frac{q_0}{A_0},s) = \frac{q}{A} \iff \delta_{G'}(q_0,s) = q \wedge \delta_{H'}(A_0,P_{\Sigma_o}(s)) = A.$$
\qed
\end{proof}

\subsection{Time complexity analysis of the \texttt{REFINE} algorithm}
We tie the asymptotic time complexity of the \texttt{REFINE} algorithm to the size of the automaton $G\ ||\ H$ with $N$ states, an alphabet of size $|\Sigma|$, and, because it is a DFA, an upper bound of $N \times |\Sigma|$ transitions.

Line~\ref{alg:denoteByM} is a memory operation where we denote the input $G\ ||\ H$ by $M$. This happens in constant time. Line~\ref{alg:initializeAr} is also a memory operation, where the set of states to remove from at most $N$ estimates are initialized to $\emptyset$. The complexity class for this operation is $\mathcal{O}(N)$. Line~\ref{alg:deltaToDeltaQ} is another memory operation where we denote the input $\Delta$ by $\Delta_Q$. This happens in constant time.

The first loop, starting at Line~\ref{alg:buildDeltaQStart}, will in the worst case examine every transition in $G\ ||\ H$ once. Practically speaking, transitions could be stored in a sorted list with all uncontrolled transitions at the start of the list. For each transition, at most four constant time operations are performed: verifying if the transition is uncontrolled; adding the origin state to $\Delta_Q$; and removing the destination state from both $\Delta_Q$ and $M$. The complexity class for this loop is $\mathcal{O}(N \times |\Sigma|)$. 

The second loop, starting at Line~\ref{alg:checkInaccessibleStatesStart}, will examine every state in $G\ ||\ H$ once. For each state, at most two constant time operations are performed: checking if the state is inaccessible; and adding the first element of the state to $A_r(A)$. The complexity class for this loop is $\mathcal{O}(N)$.

The third loop, starting at Line~\ref{alg:accessibleStart} will examine every state in $G\ ||\ H$ once. For each state, exactly one constant time operation is performed: relabelling the state's estimate by removing the previously stored set $A_r(A)$ from the current estimate. The complexity class for this loop is $\mathcal{O}(N)$.

Line~\ref{alg:onlyAccessible} requires returning the accessible portion of $M$. Practically, the inaccessible components of $M$ can be removed during the second loop, and so this is a constant time operation.

Taken all together, the complexity class for \texttt{REFINE} is $\mathcal{O}((N \times 3) + (N \times |\Sigma|))$, which simplifies to $\mathcal{O}(N \times |\Sigma|)$. If we consider that the number of states in an automaton is generally larger than the number of events, this further simplifies to $\mathcal{O}(N)$.

\subsection{Discussion}
The subobserver relationship is a more general form of the subautomaton relationship. We showed that the subobserver relationship captures how the observer automaton changes to reflect evolutions in the plant. Specifically, the subobserver relationship allows the parallel composition of a plant and an observer automaton for that plant, $G\ ||\ H$, to be refined without recomputing the observer. Instead, our method removes transitions in the parallel composition if and only if they must be disabled by supervisory control and updates the observer's estimates. The \texttt{REFINE} algorithm implements this procedure and its asymptotic time complexity is linear in the number of states of $G\ ||\ H$.

Although the \texttt{REFINE} algorithm was developed within the framework of the opacity control problem, it more generally produces the joint behaviour of plant and observer as the plant evolves over time. In the remainder of this paper we consider the opacity control problem exclusively, but we note that this approach is applicable to any DES problem where evolving discrete-event processes are observed including online control, dynamic discrete-event systems, and decentralized control.

\section{Synthesizing Opacity-Enforcing Supervisors}
Our opacity control problem formulation is similar to that from other works (see Problem~\ref{prob:opacitycontrolproblem}). We assume that all controllable events are observable, $\Sigma_c \subseteq \Sigma_s$, and that the adversary observes a subset of the events that the supervisor observes, $\Sigma_a \subseteq \Sigma_s$. These assumptions allow us to use Dubreil et al.'s reduction of the general opacity control problem to the opacity control problem with full observation, i.e.\ $\Sigma_s = \Sigma$~\cite[Proposition~5]{Dubreil2010}. We assume a supervisor-aware adversary, which implies that the adversary's view will evolve with $G' = \mathcal{S}/G$. For simplicity we enforce current-state opacity.
\begin{figure}[!htb]
\centering
\subfloat[The na\"{i}ve method]{
\includegraphics[width=0.4\textwidth]{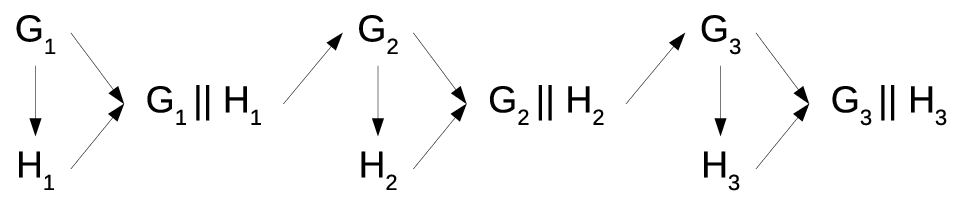}
\label{fig:naive}}\qquad
\subfloat[Condensed state estimates~\cite{Dubreil2010}]{
\includegraphics[width=0.4\textwidth]{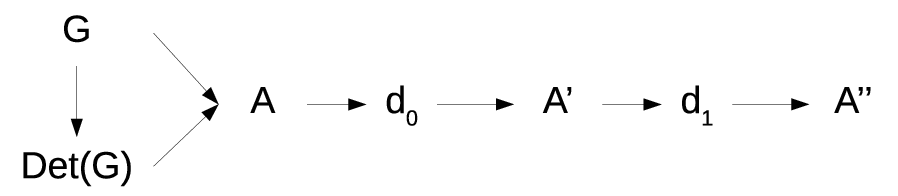}
\label{fig:dubreilmethod}}\\
\subfloat[Augmented I-observer~\cite{Tong2018a}]{
\includegraphics[width=0.4\textwidth]{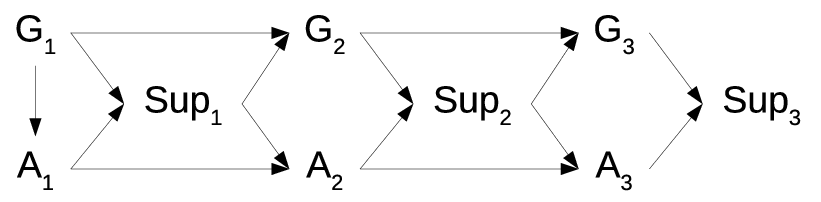}
\label{fig:tongmethod}}\qquad
\subfloat[The \texttt{SYNTHESIZE} algorithm]{
\includegraphics[width=0.4\textwidth]{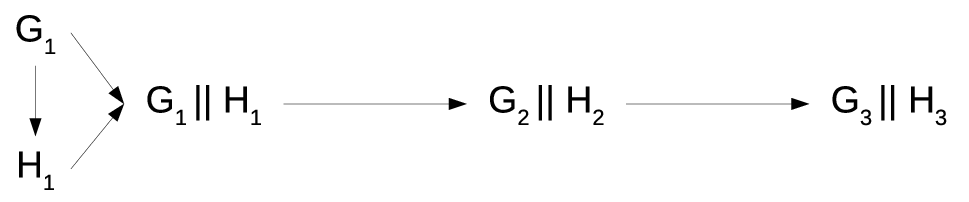}
\label{fig:proposed}}
\caption{Comparing the calculations and data structures for each approach.}
\end{figure}

Recalling approaches to the opacity control problem, the na{\"{i}}ve, computationally inefficient solution is to apply supervisory control and calculate observers iteratively (Figure~\ref{fig:naive})~\cite{Dubreil2010}. Condensed state estimates (Figure~\ref{fig:dubreilmethod}) require alternately constructing partial maps and automata; since the only states in a condensed state estimate are those that were reached by the last event observable to the adversary, ``loosing traces'' are calculated at each iteration to determine which states will reveal a secret~\cite{Dubreil2010}. The augmented I-observer (Figure~\ref{fig:tongmethod}) solves a more general form of the problem but requires calculating the supremal G-opaque sublanguage, synthesizing a supervisor and refining both the plant and augmented I-observer at each step. Tong et al.\ avoid this complication by restricting themselves to the case of an adversary who is unaware of the supervisor~\cite{Tong2018a}.

Our method uses the subobserver property to calculate $G'\ ||\ H'$ directly from $G\ ||\ H$, refining the parallel composition of plant and adversary view (Figure~\ref{fig:proposed}). Advantages of this method include: not producing any intermediary structures or languages; capturing the adversary's beliefs better than condensed state estimates, thus not requiring ``loosing traces'' to be calculated; and avoiding unnecessary calculations compared to the augmented I-observer with a supervisor-aware adversary.

\subsection{The \texttt{SYNTHESIZE} algorithm} \label{sec:proposedAlgorithm}
As in Section~\ref{sec:updatingObservers}, we use the parallel composition of $G$ and $H := T_{\Sigma_a}(G)$ to track the plant's evolving structure along with the adversary's beliefs about the plant. Importantly, we mark states in $H$ to capture the scenarios in which no non-secret state can be confused by the adversary with a true secret state. With a traditional observer automaton, a state is marked in $H$ when one of the states that it contains is itself marked in $G$; in our work, by contrast, a state in $H$ is marked when all the states it contains are themselves marked in $G$, as described in Section~\ref{sec:observingDES}. Thus, in our work, if $H$ -- the observer -- reaches a marked state, then a string has occurred in $G$ -- the plant -- that allows the adversary to be sure that a secret state has been reached. This marking carries through to $G\ ||\ H$, where marked states must be made inaccessible to produce an opacity-enforcing supervisor. 
\begin{algorithm}[ht]
\KwData{Plant $G$; secret states $Q_m$ and non-secret states $Q \setminus Q_m$; alphabet visible to the adversary $\Sigma_a \subseteq \Sigma$ and alphabet controllable by supervisor $\Sigma_c \subseteq \Sigma$.}
\KwResult{Supervisor $\mathcal{S}$ such that $\mathcal{S}/G$ is current-state opaque with respect to $Q_m$ and $\Sigma_a$.}
\tcc{Create the observer automaton and the parallel composition.}
$H = T_{\Sigma_a}(G)$\; \label{alg:computeGa}
$M =\ G\ ||\ H$\; \label{alg:gCrossga}
\While{$M$ has marked states \label{alg:markedStatesBegin}}
{
\tcc{$\Delta$ is the list of states to make inaccessible in $M$, i.e.\ the states that cause opacity to not be enforced.}
Set $\Delta$ equal to the set of marked states in $M$\;
\tcc{Refining $M$ reflects that the adversary changes its view of the plant based on its ability to reason about the policy enforced by the supervisor.}
$M =\ $\texttt{REFINE}($M,\Delta$)\;
\label{alg:refineComposition}
\tcc{Mark any states whose new estimates contain only secret states.}
\ForEach{State $\frac{q}{A}$ in $M$ \label{alg:checkAdversaryBeliefs}}
{
\If{$A \subseteq Q_m$}
{
Mark the state $\frac{q}{A}$ in $M$\;
}
}
} \label{alg:markedStatesEnd}
\tcc{If $M$ is an empty automaton, then there is no valid supervisor for $G$ that is able to enforce current-state opacity with respect to the adversary's observations $\Sigma_a$.}
$\mathcal{S} = M$\;
\caption{The \texttt{SYNTHESIZE} algorithm}
\label{alg:proposed}
\end{algorithm}

The \texttt{SYNTHESIZE} algorithm (Algorithm~\ref{alg:proposed}) enforces current-state opacity by applying supervisory control as long as the controlled plant composed with the adversary's estimations of the plant contains marked states (Lines~\ref{alg:markedStatesBegin}-\ref{alg:markedStatesEnd}), the adversary's view of the plant is updated (Line~\ref{alg:refineComposition}), and the new composition of plant behaviour and adversarial belief is checked for current-state opacity (Lines~\ref{alg:checkAdversaryBeliefs}-\ref{alg:markedStatesEnd}).

\subsection{Worked example}
We will demonstrate how the \texttt{SYNTHESIZE} algorithm works for a concrete example, beginning with the plant and the initial parallel composition (Figure~\ref{fig:CUMCex1a}). This plant has two marked states, $6$ and $10$, an event that the supervisor cannot control, $\beta$, and events that the adversary cannot observe, $\alpha$ and $\gamma$.
\begin{figure}[htb]
\centering
\subfloat[The original plant, $G$]{
\includegraphics[width=0.4\textwidth]{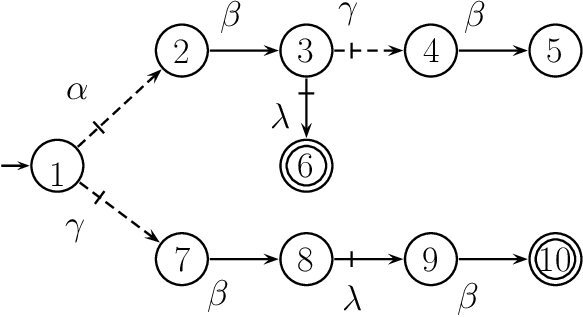}}\qquad
\subfloat[Initial parallel composition, $G\ ||\ H$]{
\includegraphics[width=0.45\textwidth]{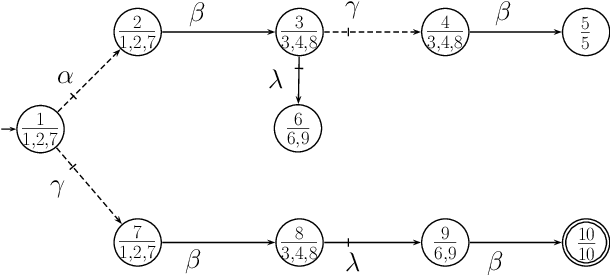}}
\caption{The initial plant and parallel composition for the worked example.}
\label{fig:CUMCex1a}
\end{figure}

The initial parallel composition has a single marked state, $\frac{10}{10}$, which is added to $\Delta$ and passed to \texttt{REFINE}. In order to prevent the state $\frac{10}{10}$ from being reached, we must also mark the state $\frac{9}{6,9}$ in $M$ since it leads uncontrollably to our marked state. Removing these two states from $M$ implies that when the supervisor is in state $\frac{8}{3,4,8}$ it will disable the event $\lambda$ via supervisory control since the transition $\langle \frac{8}{3,4,8}, \lambda, \frac{9}{6,9} \rangle_M$ is no longer defined. Removing these states from $M$, however, means that our supervisor will not allow the plant $G$ to reach state $9$, which will affect the adversary's estimates. Specifically here, when the plant is in state $6$, the adversary will no longer confuse the string in $G$ that led to state $6$ with a string in $G$ that leads to state $9$ since the latter is now inaccessible in the plant. This leads to the new parallel composition in Figure~\ref{fig:CUMCsync2}, which also has one marked state, $\frac{6}{6}$.

Removing the state $\frac{6}{6}$ from $M$ produces the parallel composition in Figure~\ref{fig:CUMCsync3}. This automaton has no marked states, indicating that it can be used as a supervisor to enforce current-state opacity in the plant $G$.
\begin{figure}[htb]
\centering
\subfloat[The second parallel composition, $G'\ ||\ H'$]{
\includegraphics[width=0.45\textwidth]{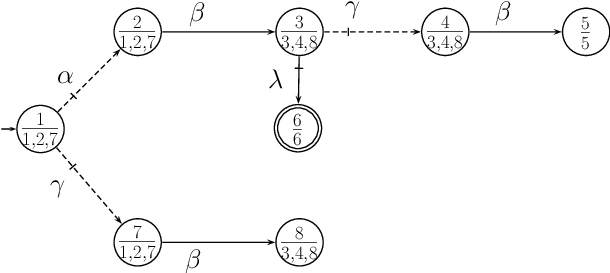}
\label{fig:CUMCsync2}}\qquad
\subfloat[The third parallel composition, $G''\ ||\ H''$]{
\includegraphics[width=0.45\textwidth]{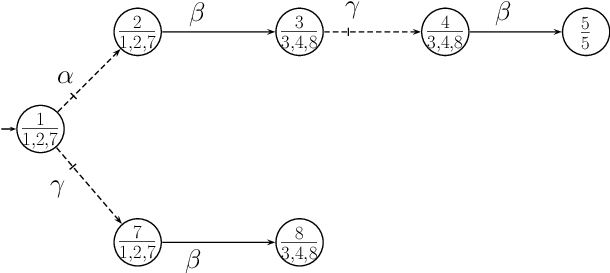}
\label{fig:CUMCsync3}}
\caption{The parallel compositions computed during the worked example.}
\label{fig:CUMCex1b}
\end{figure}

\subsection{The case of state splitting} \label{sec:stateSplitting}
As remarked in Section~\ref{sec:relatingObservers}, we are not guaranteed that applying supervisory control to $G$ will result in a subautomaton of $G$. For example, the observer automaton in Figure~\ref{fig:specNotSubAutEx} illustrates how although the secret states 8 and 9 are repeatedly visited, the adversary is only able to determine that the plant is in a secret state once it sees a fourth $a$. This means that a supervisor can allow state 8 to be visited up to three times and must therefore enact different control patterns for state 8 depending on how many times the plant has visited state 8.
\begin{figure}[b]
\centering
\subfloat[Plant $G$, which requires state splitting.]{
\includegraphics[height=3cm]{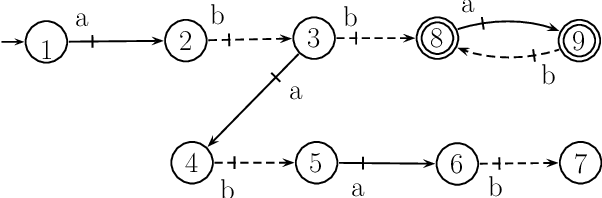}}\\
\subfloat[Observer automaton, $H$.]{
\includegraphics[height=1.25cm]{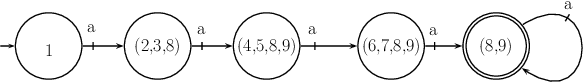}}
\caption{States 8 and 9 in $G$ are secret and only $a$ is observable to the adversary. The observer automaton shows that after observing a fourth $a$, an adversary will know the plant is in a secret state. Note that only $\{8,9\}$ is marked in $H$ because this is the only state whose constituent states are all marked in $G$.}
\label{fig:specNotSubAutEx}
\end{figure}

State splitting is required for opacity enforcement when a state can be visited in the prefixes of strings leading to secret states but only a finite number of times, requiring the supervisor to know which string led to the state for a particular visit~\cite[p.~141]{Cassandras2008}. Note that if a state can be visited an infinite number of times then the supervisor can always enact the same control pattern. Happily, this state splitting occurs implicitly in Line~\ref{alg:gCrossga} of the \texttt{SYNTHESIZE} algorithm when the first parallel composition is constructed (Figure~\ref{fig:specNotSubAutExSync}). Because a state that requires splitting can only be safely visited a finite number of times, it is possible to transform any DFA, $G$, into a language-equivalent DFA, $\hat{G}$, by replicating the original state for each time that it can be safely visited.
\begin{figure}[htb]
\centering
\includegraphics[width=0.75\textwidth]{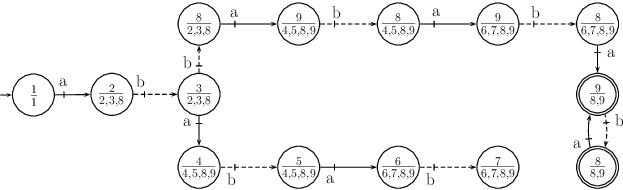}
\caption{Continuing from Figure~\ref{fig:specNotSubAutEx}, the first parallel composition inherently constructs the needed, language-equivalent automaton to represent the plant.}
\label{fig:specNotSubAutExSync}
\end{figure}

In fact, taking only the first element of the states in $G\ ||\ H$'s state transition diagram produces $\hat{G}$'s state transition diagram. Following this initial transformation, we are then guaranteed that $G'$ will be a strict subautomaton of $\hat{G}$ and the \texttt{SYNTHESIZE} algorithm can proceed from there. We note that this state splitting is not an artifact of our proposed method: any method for producing a minimally restrictive opacity-enforcing supervisor from the original plant $G$ must be able to reason about the number of times that such states have been reached. The structure $G\ ||\ H$ embeds this information by splitting states the necessary number of times; it is therefore the smallest memoryless structure that can be used to solve the opacity control problem in the context of a supervisor-aware adversary.

\subsection{Proof of correctness for the \texttt{SYNTHESIZE} algorithm}
First, the \texttt{SYNTHESIZE} algorithm is guaranteed to terminate. As with the \texttt{REFINE} algorithm, termination is guaranteed because the automaton $G\ ||\ H$ has a finite number of states and the \texttt{SYNTHESIZE} algorithm makes states inaccessible until $M$ enforces opacity or is the null automaton. Specifically, if $M$ has any marked states then the algorithm removes states -- never adding states -- to make the marked states inaccessible. Since $M$ has a finite number of states, and the algorithm terminates if $M$ is the null automaton, the algorithm is guaranteed to terminate.

Second, we prove the correctness of \texttt{SYNTHESIZE} through Theorem~\ref{the:correctnessOfSynthesize}. 
\begin{theorem} \label{the:correctnessOfSynthesize}
Given a DFA $G = (Q,\Sigma,\delta_G,q_0,Q_m)$, set of secret states $Q_m$, set of events visible to the adversary $\Sigma_a \subseteq \Sigma$, and set of events controllable by the supervisor $\Sigma_c \subseteq \Sigma$. Then the supervisor $\mathcal{S}$ produced by the \texttt{SYNTHESIZE} algorithm is a correct and maximally-permissive opacity-enforcing supervisor for $G$.
\end{theorem}
\begin{proof}
We begin by proving by contradiction that the supervisor $\mathcal{S}$ correctly enforces opacity for $G$. We assume that $\mathcal{S}$ does not enforce opacity for $G$. In this case, at least one of the states in $\mathcal{S} = M$ must be marked, reflecting that all of the states in the adversary's estimate of the plant are secret states in $G$. If this was the case, however, then the while loop at Lines~\ref{alg:markedStatesBegin} to \ref{alg:markedStatesEnd} would have executed again and removed the last controllable transition along each path to this state, rendering it inaccessible. Therefore, when the \texttt{SYNTHESIZE} algorithm terminates, $\mathcal{S} = M$ has no marked states and it therefore correctly enforces opacity for $G$.

Next, we prove by contradiction that the supervisor $\mathcal{S}$ is maximally-permissive while enforcing opacity for $G$. Assume that there exists a supervisor $\mathcal{S}'$ that correctly enforces opacity for $G$ while being more permissive than supervisor $\mathcal{S}$. This implies that there is a transition in $G\ ||\ H$, $\delta^\star$, that $\mathcal{S}'$ enables and that $\mathcal{S}$ disables. But the only transitions that $\mathcal{S}$ disables in $G\ ||\ H$ are those that lead directly to a marked state or to a state which leads uncontrollably to a marked state, so $\delta^*$ must either lead directly or uncontrollably to a marked state. Therefore, $\mathcal{S}'$ does not correctly enforce opacity for $G$.
\qed
\end{proof}

\subsection{Time complexity analysis of the \texttt{SYNTHESIZE} algorithm}
Although the \texttt{SYNTHESIZE} algorithm presents a straightforward way of seeing how we can iteratively generate a list of transitions, $\Delta$, and then refine the parallel composition of plant and adversary's view, its asymptotic time complexity is in a higher complexity class than other methods in the literature. We therefore present Algorithm~\ref{alg:proposedTimeEfficient}, which embeds the \texttt{REFINE} algorithm into the \texttt{SYNTHESIZE} algorithm in a computationally efficient manner. We tie the asymptotic behaviour of Algorithm~\ref{alg:proposedTimeEfficient} to the input of interest, the automaton $G$ with $|Q|$ states and an alphabet of size $|\Sigma|$.
\begin{algorithm}[htb]
\KwData{Plant $G$; secret states $Q_m$ and non-secret states $Q \setminus Q_m$; alphabet visible to the adversary $\Sigma_a \subseteq \Sigma$ and alphabet controllable by supervisor $\Sigma_c \subseteq \Sigma$.}
\KwResult{Supervisor $\mathcal{S}$ such that $\mathcal{S}/G$ is current-state opaque with respect to $Q_m$ and $\Sigma_a$.}
\tcc{Initialize the algorithm.}
$H = T_{\Sigma_a}(G)$ \label{alg:startInitialize}\;
$M =\ G\ ||\ H$\;
Produce $\mathcal{H}$, a hash table with adversary estimates as keys and states of $M$ as values \label{alg:endInitialize}\;
\ForEach{marked state in $M$, $m = \frac{q}{A}$ \label{alg:foreachMarkedStateBegin}} 
{
\ForEach{transition leading into $m$, $t = \langle m', \sigma, m \rangle_M$ \label{alg:foreachTransitionsBegin}}
{
\uIf(the transition $t$ is uncontrollable){$\sigma \in \Sigma_{uc}$}
{Mark $m'$ in $M$\;}
}\label{alg:foreachTransitionsEnd}
Remove the state $m$ from $M$\; \label{alg:markedStateRemoved}
Go to the key $A$ in $\mathcal{H}$ and replace it with $A \setminus \{q\}$\; \label{alg:updateHashTableKey}
\If{$A \setminus \{q\} \subseteq Q_m$ \label{alg:markNewStatesBegin}} 
{Mark all states in the hash table slot indexed by $\mathcal{H}(A \setminus \{q\})$ \label{alg:markStatesInSlot}\;
}\label{alg:markNewStatesEnd}
} \label{alg:foreachMarkedStateEnd}
$\mathcal{S} = M$\;
\caption{An computationally efficient implementation of the \texttt{SYNTHESIZE} algorithm}
\label{alg:proposedTimeEfficient}
\end{algorithm}

Initializing the algorithm, Lines~\ref{alg:startInitialize}-\ref{alg:endInitialize}, requires constructing the adversary's view, $H = T_{\Sigma_a}(G)$, and the parallel composition of $G\ ||\ H$. The time complexity of this process is established to be $\mathcal{O}(|Q| \times 2^{|Q|})$. The production of the hash table $\mathcal{H}$ requires an operation for each $|Q|$ state in $M$ as well and in practice it can be constructed alongside the parallel composition.

The outer loop iterates through the marked states in $M$, Lines~\ref{alg:foreachMarkedStateBegin} to \ref{alg:foreachMarkedStateEnd}. This loop will run a variable number of times, but we will be able to assess the loop's complexity by establishing upper bounds for each of the lines inside this loop. First, the loop responsible for inspecting transitions in $M$, Lines~\ref{alg:foreachTransitionsBegin}-\ref{alg:foreachTransitionsEnd}, visits each transition in $M$ at most once and at most one constant time operation is performed. This is guaranteed because each visited transition is visited because its destination state is marked. Because that destination state will be removed at Line~\ref{alg:markedStateRemoved} of the current iteration, the transition in question will not be visited again. Lines~\ref{alg:foreachTransitionsBegin}-\ref{alg:foreachTransitionsEnd} therefore have an asymptotic time complexity over the algorithm's whole run time of $\mathcal{O}(|Q| \times 2^{|Q|} \times |\Sigma|)$. Second, the lines responsible for processing marked states in $M$, Lines~\ref{alg:markedStateRemoved}-\ref{alg:markNewStatesEnd} consists of up to four operations. Removing the state $m$ from $M$ has a constant time complexity; because each state in $M$ can be removed at most once this line has an asymptotic time complexity over the algorithm's whole run time of $\mathcal{O}(|Q| \times 2^{|Q|})$. Updating the hash table key similarly has a constant time complexity and can only occur once per state in $M$. This again leads to an asymptotic time complexity over the algorithm's whole run time of $\mathcal{O}(|Q| \times 2^{|Q|})$. Finally, checking if the new estimate consists of solely secret states (Line~\ref{alg:markNewStatesBegin}) takes at most $|Q_m|$ checks and then marking states in the hash table slot (Line~\ref{alg:markStatesInSlot}) takes constant time by state. Since every state can only be marked once over the algorithm's whole run time, this again leads to an asymptotic time complexity of $\mathcal{O}(|Q| \times 2^{|Q|})$.

Taken together, the asymptotic time complexity of Algorithm~\ref{alg:proposedTimeEfficient} is $\mathcal{O}(|Q| \times 2^{|Q|} \times |\Sigma|)$. Since in many problems the size of the alphabet does not grow and is much smaller than the size of the state space, we can consider $|\Sigma|$ as a constant and reduce this time complexity to $\mathcal{O}(|Q| \times 2^{|Q|})$. Our result is in keeping with the exponential time required for computing an observer automaton and we note that there is no algorithm for verifying current-state opacity for a system $G$ whose time complexity is polynomial in the number of states in $G$~\cite{Wu2013}.

\subsection{Discussion}
We have demonstrated how our method for refining parallel compositions of plant and observer can be used to synthesize an opacity-enforcing supervisor. The \texttt{SYNTHESIZE} algorithm implements this procedure in a direct manner and Algorithm~\ref{alg:proposedTimeEfficient} trades a small increase in memory usage to reduce the asymptotic time complexity to $\mathcal{O}(|Q| \times 2^{|Q|})$ where $|Q|$ is the number of states in the original plant.

We showed two examples to demonstrate how our algorithm works. The first example highlighted the interplay between enacting supervisory control and enforcing opacity when faced with a supervisor-aware adversary: multiple refinements of the automaton $G\ ||\ H$ may be required before an output is produced that can be used as an opacity-enforcing supervisor. The second example showed how our approach inherently addresses splits states when a secret state can be safely visited only a finite number of times without revealing the system secret. Recalling the two methods in the literature that are most similar to ours, we note that Algorithm~\ref{alg:proposedTimeEfficient} offers improvements on both approaches. Although each method has the same asymptotic time complexity,~\cite{Tong2018a}, our method produces only a single structure at each step: the parallel composition of plant and adversary estimate.

The method of condensed state estimates, by contrast, produces both a partial map and an automaton at every step~\cite{Dubreil2010}. Additionally this method requires ``loosing paths'' to be assessed in order to verify the opacity property because condensed state estimates account only for states that the plant might have reached if the last observed event is the last event that occurred in the plant. Our method does away with the need to reason about these ``loosing paths,'' which produces two advantages. First, opacity is immediately verified by the lack of marked states in the output which permits an opacity-enforcing supervisor to be easily verified by an independent party. Second, the output's state labels are semantically meaningful: if a trace occurs in the system then this trace leads to a state in the parallel composition whose label is the current state of the system and the states that the observer believes the system could be in. This information can be used by an engineer or a computer program to make decisions about the system.

The augmented I-observer method, solves a more general problem by allowing $\Sigma_s$ and $\Sigma_a$ to be incomparable~\cite{Tong2018a}. Balanced against this, it focuses on the case of a single iteration to enforce opacity against a non-supervisor-aware adversary. Because this method calculates the supremal G-opaque sublanguage of the plant's language and a supervisor to enforce this language as a specification,~\cite{Tong2018a}, this results in unnecessary calculations when a a supervisor-aware adversary makes multiple iterations necessary. Our method makes the realistic assumption that $\Sigma_a \subseteq \Sigma_s$ and addresses the case of a supervisor-aware adversary with only a single structure computed at each step.

\section{Conclusion}
The opacity control problem has been of interest in the DES literature for many years and has been addressed a number of times. It is a complex problem whose straightforward solution is computationally expensive. Compared with our method, the solutions to the opacity control problem in the literature variously require additional computations to deal with the problem in its general form or use structures that do not intuitively align with the opacity control problem.

This paper introduced and described the subobserver relationship. This relationship is analogous to the subautomaton relationship and links the observers of a plant whose structure evolves. We demonstrated the usefulness of the subobserver relationship by using it as the basis of an algorithm to solve the opacity control problem and we believe that this relationship has broader application in DES research including online control, dynamic discrete-event systems, and decentralized control.

Time complexity analysis shows that our algorithms are computationally efficient, with Algorithm~\ref{alg:proposedTimeEfficient} matching the asymptotic time complexity of previous methods in the literature. We also demonstrated that our algorithms intuitively solve the opacity control problem under reasonable assumptions. Future research applying these ideas may include relaxing the requirement that $\Sigma_c \subseteq \Sigma_s$ for the opacity control problem and applying the subobserver relationship to the iterative refinements necessary when a single supervisor enacts online control or several agents enact decentralized control over a plant.

\section*{Conflicts of Interest}
The authors declare that they have no conflict of interest.

\begin{acknowledgements}
The authors acknowledge that Queen's University is situated on traditional Anishinaabe and Haudenosaunee Territory.

The authors would like to thank Richard Ean and Bryony Schonewille of the Queen's Discrete-Event Systems Lab as well as the reviewers for their thoughtful and constructive comments. This research was supported by the Natural Sciences and Engineering Research Council of Canada as well as the Queen's University Faculty of Engineering and Applied Science. RHM held a Walter C. Sumner Memorial Fellowship during the conduct of this work.

All figures with automata were produced using the Integrated Discrete-Event Systems Tool,~\cite{Rudie2006}, which is freely available online at \texttt{https://github.com/krudie/IDES}.
\end{acknowledgements}

\appendix
\section{Omitted Proofs}
\begin{proof}[Proof of Lemma~\ref{lem:stateInGH}]
We begin by noting that $L(G\ ||\ H) = L(G)$ and that because $\Sigma_o \subseteq \Sigma$, the alphabet for $G\ ||\ H$ is $\Sigma$ as well.

$\implies$. If a state $\frac{q}{A}$ is accessible in $G\ ||\ H$, then there exists a string $s \in L(G\ ||\ H)$ such that $\delta_{G\ ||\ H}(\frac{q_0}{A_0},s) = \frac{q}{A}$. By Definition~8, this implies that that $\delta_{G}(q_0,s) = q$ and that $\delta_{H}(A_0,P_{\Sigma_o}(s)) = A$. Since $\delta_{H}(A_0,P_{\Sigma_o}(s)) = A$, we know that every string that shares a projection with $s$ leads to a state in $A$ and that every state in $A$ is reached by a string that shares a projection with $s$. So we can conclude that $$(\exists s \in L(G))(\delta_G(q_0,s) = q)(\{q' \in Q\ |\ (\exists s' \in L(G)\ |\ P_{\Sigma_o}(s) = P_{\Sigma_o}(s'))(\delta_G(q_0,s')=q')\} = A).$$

$\impliedby$. We are given that $$(\exists s \in L(G))(\delta_G(q_0,s) = q)(\{q' \in Q\ |\ (\exists s' \in L(G)\ |\ P_{\Sigma_o}(s) = P_{\Sigma_o}(s'))(\delta_G(q_0,s')=q')\} = A).$$ We therefore know that $\delta_{G}(q_0,s) = q$ and that $\delta_{H}(A_0,P_{\Sigma_o}(s)) = A$. By Definition~8, this implies that $\delta_{G\ ||\ H}(\frac{q_0}{A_0},s) = \frac{q}{A}$ and therefore $\frac{q}{A}$ is accessible in $G\ ||\ H$.
\qed
\end{proof}

\begin{proof}[Proof of Lemma~\ref{lem:produceGPrimeSubautomaton}]
According to the definition, $G'$ is a subautomaton of $G$, denoted by $G'\ \sqsubseteq\ G$, if $\delta_{G'}(q'_0,s) = \delta_G(q_{0},s)\quad \forall\ s \in L(G').$

Denote the first elements of states remaining in $G\ ||\ H$ as $G'$. Then we know that $$(\forall\ s \in L(G'))(\delta_{G'}(q_0,s) = q \implies \delta_{G}(q_0,s) = q)$$ since any first element of a state that survives in $G\ ||\ H$ was originally in $G$. This means that any string that occurs in $G'$ leads to the same state that it did in $G$ and therefore that $\delta_{G'}(q'_0,s) = \delta_G(q_{0},s)\quad \forall\ s \in L(G').$

We conclude that $G'$ is a subautomaton of $G$.
\qed
\end{proof}

\bibliographystyle{spmpsci}
\bibliography{Subobserver.bib}
\end{document}